\documentclass[aps,prl,reprint,superscriptaddress,floatfix,longbibliography]{revtex4-2}

\usepackage[unicode=true,pdfusetitle, bookmarks=true,bookmarksnumbered=false,bookmarksopen=false, breaklinks=false,pdfborder={0 0 0},backref=false,colorlinks=false,pagebackref=false]{hyperref}

\hypersetup{colorlinks,linkcolor=myurlcolor,citecolor=myurlcolor,urlcolor=myurlcolor}
\usepackage{graphics,epstopdf,graphicx,amsthm,amsmath,amssymb,braket,colortbl,color,bm,framed,mathrsfs}
\usepackage{float}
\usepackage{enumerate}
\usepackage[inline]{enumitem}
\usepackage{cleveref}

\usepackage{booktabs} % Required for better horizontal rules in tables

\definecolor{myurlcolor}{rgb}{0,0,0.9}

\DeclareMathOperator{\trace}{Tr}

\newcommand{\Ptr}[2]{\trace_{#1}\Pa{#2}}
\newcommand{\Tr}[1]{\Ptr{}{#1}}

\newcommand{\Pa}[1]{\left[#1\right]}

\newcommand{\Supp}{\operatorname{Supp}}

\theoremstyle{plain}
\newtheorem{thm}{Theorem}
\newtheorem{lemma}[thm]{Lemma}
\newtheorem{proposition}[thm]{Proposition}
\newtheorem{cor}[thm]{Corollary}

\theoremstyle{definition}
\newtheorem{Def}[thm]{Definition}

\newenvironment{proofsketch}{\begin{proof}[Proof sketch]}{\end{proof}}
\usepackage{pifont}

\newcommand{\FG}{\mathsf{FG}}
\newcommand{\FGC}{\mathsf{FGC}}

\newcommand{\Herm}{\operatorname{Herm}}
\newcommand{\ad}{\operatorname{ad}}

\definecolor{yrteal}{rgb}{0.0,0.45,0.45}

\begin{document}

\title{Convex-Gaussianity of fermionic Gibbs states in perturbation theory}
\author{Kaifeng Bu}
\email{bu.115@osu.edu}
\affiliation{Department of Mathematics, The Ohio State University, Columbus, Ohio 43210, USA}
\affiliation{Department of Physics, Harvard University, Cambridge, Massachusetts 02138, USA}
\author{Yuanjie Ren}
\email{yuanjie@mit.edu}
\affiliation{Department of Physics, Massachusetts Institute of Technology, Cambridge, Massachusetts 02139, USA}
\date{\today}

\begin{abstract}
We study the structure of Gibbs states in weakly perturbed interacting fermionic systems. First, for a sparse Hamiltonian $H=H_0+V$ with a quadratic term $H_0$ and a non-quadratic perturbation $V$ of scale $\epsilon$, we show that the Gibbs state $\rho_{\beta}$ decomposes into a convex combination of Gaussian states whenever the inverse temperature satisfies $\beta \le O(\log(1/\epsilon))$. Moreover, we prove that this bound is asymptotically tight by establishing that $\beta \le \Theta(\log(1/\epsilon))$ is necessary for certain sparse Hamiltonians. This general framework applies directly to the weak-coupling (small-$\vert{}U\vert{}$) regime of the Fermi--Hubbard model with hopping $t$ and on-site interaction $U$ on any graph of maximum degree $D$. Complementarily, in the strong-coupling (small-$\vert{}t\vert{}$) regime, we show that the Gibbs state remains convex-Gaussian up to $\beta \le O\big(\vert{}U\vert{}^{-1}\log(\vert{}U\vert{}/(D\vert{}t\vert{}))\big)$, revealing a mechanism for convex-Gaussianity distinct from the weak-coupling setting.

\end{abstract}

\maketitle

\section{Introduction}

For a quantum many-body Hamiltonian $H$ at inverse temperature $\beta = 1/T$, the thermal equilibrium state is described by the Gibbs state~\cite{PhysRevX.4.031019},
$\rho_\beta(H)=e^{-\beta H}/\Tr{e^{-\beta H}}$. The Gibbs state determines finite-temperature expectation values and fundamental physical properties, including correlation functions, response functions, free energies, and phase behavior~\cite{PhysRevX.4.031019, D_Alessio_2016}. Consequently, preparing, sampling, and learning Gibbs states constitute a central objective in both classical and quantum many-body simulation~\cite{PhysRevLett.103.220502, Temme2011, chowdhury2016quantumalgorithmsgibbssampling, Brandão2019, PRXQuantum.4.040201, Chen2025, Bergamaschi_2024, Chen_2025, chen2025quantumgibbsstateslocally}. For interacting fermionic systems, however, direct computation is notoriously challenging due to the exponential growth of the Hilbert space and the severe negative sign problem inherent to generic quantum Monte Carlo methods~\cite{PhysRevLett.94.170201, Li_2019}. It is therefore essential to identify interacting regimes in which fermionic Gibbs states admit a simpler, tractable structure and to characterize their classical or quantum simulability within the broader complexity phase diagram.

Quadratic fermionic Hamiltonians form an important, tractable class of models whose Gibbs states at all temperatures are fermionic Gaussian states. These thermal states are fully determined by two-point correlations, obey Wick's theorem, and admit efficient simulation via covariance-matrix methods~\cite{SciPostPhysLectNotes.54,Bravyi_2005,Bravyi_2012}. In addition, Gaussian states and operations are closely connected to matchgate quantum computation~\cite{Jozsa_2008,knill2001fermioniclinearopticsmatchgates,PhysRevA.65.032325}.
Fermionic Gaussian states and operators are preserved under fermionic convolution~\cite{lyu2024fermionicgaussiantestingnongaussian} and admit diagrammatic representations within a graphical calculus for fermionic tensors~\cite{ren2025graphicalcalculusfermionictensors}.
In quantum resource theory,  convex combinations do not generate non-free resources. Hence, 
we refer to a state as \emph{convex-Gaussian} if it is a convex combination of fermionic Gaussian states~\cite{Melo_2013,PhysRevA.90.062329}. This property provides a natural extension of free-fermion physics and is closely related to questions of fermionic linear optics, classical simulation, and non-Gaussian computational resources~\cite{Bravyi_2005,PhysRevA.90.020302,PhysRevLett.123.080503,Dias2024classicalsimulation,PRXQuantum.3.020328}.

In general, the Gibbs states of interacting fermionic Hamiltonians are not necessarily convex combinations of Gaussian states. 
At high temperatures $T \to \infty$, i.e., $\beta \to 0$, the Gibbs state of any Hamiltonian simplifies to a convex combination of Gaussian states. However,
determining the precise temperature threshold at which this property holds remains a non-trivial problem. 
More generally, high-temperature Gibbs states of short-range local systems
can exhibit decay of correlations and approximate Markov structure under suitable conditions
~\cite{PhysRevX.4.031019,
PhysRevLett.124.220601}.
Recently, Ramkumar et al. show that
Gibbs states of bounded-degree local fermionic Hamiltonians are convex-Gaussian
if the temperature is sufficiently high, with a threshold independent of system size~\cite{qprk-k7wn}.
This result is similar in spirit to work by Bakshi et al., which shows that
Gibbs states of bounded-degree local Hamiltonians are separable
if the temperature is sufficiently high~\cite{Bakshi_2024}.
In addition, for local Pauli Hamiltonians of locality $k$ and local strength $s$
that are $\epsilon$ close to commuting, the Gibbs state is proved to be a convex
combination of stabilizer states when
$\beta\lesssim \log(1/\epsilon)/(sk)$~\cite{putterman2026quantumthermalstateslook}.
These results leave
open how much farther convex Gaussianity persists when the interaction is
weak relative to a quadratic reference for fermionic systems, and what replaces that argument in
the strong-coupling Fermi--Hubbard regime.

Here we answer both questions in complementary regimes. We consider a sparse 
Hamiltonian $H=H_0+V$, where the quadratic term $H_0$ has energy scale $\mathcal E$,
and the nonquadratic terms in $V$ have locality bounded by $R$ and small relative scale $\epsilon$.
Then
its Gibbs state is a convex combination of Gaussian states 
if the inverse temperature satisfies
\begin{align}
 \beta\le
 \frac{1}{R(d+1)\mathcal E}\log\!\left(1+\frac{1}{72\epsilon}\right),
\end{align}
where $d$ is the degree of the interaction graph of the Hamiltonian.
Moreover, we construct a family of fermionic Hamiltonians
such that $\beta\leq O(\log(1/\epsilon))$ is necessary to guarantee that the Gibbs state is
a convex combination of Gaussian states. 

In addition, we apply our results to the Fermi--Hubbard model $H=tH_{\rm hop}+UV_{\rm int}$ on any finite degree-$D$ graph, with
$H_{\rm hop}$ being the hopping term
and $V_{\rm int}$ the on-site interaction term  \cite{Tasaki_1998,Arovas_2022,
Qin_2022,
CRPHYS_2018__19_6_365_0}.
This implies that 
the Gibbs state of the Fermi--Hubbard model in the weak-coupling regime (i.e., $ |U|\ll |t|$)
is a convex combination of Gaussian states if the inverse temperature satisfies
\begin{align}
    \beta\leq O(|t|^{-1}\log(|t|/|U|)).
\end{align}
Moreover, in the strong-coupling regime (i.e., $ |t|\ll|U|$), we also show that its
Gibbs state is a convex combination of Gaussian states
if the inverse temperature satisfies 
\begin{align}
    \beta\leq O(|U|^{-1}\log(|U|/|t|)).
\end{align}

\section{Preliminaries}

Consider $n$ fermionic modes with annihilation and creation operators $c_j$ and $c_j^\dagger$, $j=1,\ldots,n$,  which satisfy the canonical anticommutation relations
$\{c_j,c_k^\dagger\}=\delta_{jk}I, \{c_j,c_k\}=0, \{c_j^\dagger,c_k^\dagger\}=0.$
Define the $2n$ Majorana operators by
$ \gamma_{2j-1}=c_j+c_j^\dagger, \gamma_{2j}=-i(c_j-c_j^\dagger)$.
They are Hermitian and satisfy
$  \{\gamma_a,\gamma_b\}=2\delta_{ab}I$.

A  quadratic Hamiltonian is a real linear combination of $I$ and the operators $i\gamma_a\gamma_b$, $a<b$. For an even set $ J=\{j_1<\cdots<j_{2r}\}$,
define the normalized Majorana string $ \gamma_J:=i^r\gamma_{j_1}\cdots\gamma_{j_{2r}}.$
The Canonical Anti-commutation Relation (CAR) implies $  \gamma_J^\dagger=\gamma_J, \gamma_J^2=I.$
For the empty set, we use $\gamma_\varnothing=I$. 
Every even operator $A$ has a unique Majorana expansion
\begin{align}
      A=\sum_{\substack{J\subseteq[2n]\\ |J|\ \mathrm{even}}} a_J\gamma_J.
\end{align}
We define its Majorana coefficient norm by
\begin{align}
     \|A\|_{\gamma,1}:=\sum_J |a_J|. 
\end{align}

One important family of fermionic states is known as fermionic Gaussian states. 
They are defined as the Gibbs states of 
some quadratic Hamiltonian $H_0=\sum_{i,j}c_{ij}\gamma_i\gamma_j$, 
i.e., $\exp(-\beta H_0)/\Tr{\exp(-\beta H_0)}$. 
And the 
 fermionic Gaussian state is pure if $\beta \to \infty$, i.e., a ground state of $H_0$.
Besides, as the convex combination is a classical process which does not generate non-free resources, we always 
consider the convex combination of Gaussian states in the 
resource theory. 
Define
\begin{align}
 \FGC_n:= \bigg\{\sum_i p_i \rho_i\mid  p_i\ge 0, \rho_i \text{ is fermionic Gaussian}\bigg\}.
\end{align}
We also denote its normalized version by
$\operatorname{conv}(\FG)_n=\set{A/\Tr{A}:A\in \FGC_n}$, which represents the set of convex combinations of fermionic Gaussian states. 
\begin{lemma}[Single-string Gaussian-cone criterion]
\label{lem:single-string}
For every even set $J$ and every real $\alpha\in[-1,1]$,
\begin{align}
    I+\alpha \gamma_J\in\FGC_n.
\end{align}
Whenever $\Tr{I+\alpha \gamma_J}>0$, its normalization $ (I+\alpha \gamma_J) /\Tr{I+\alpha \gamma_J}$ belongs to $\operatorname{conv}(\FG)_n$.
\end{lemma}
  \begin{proof}
  This criterion is implicit in the recursive pinning construction of
  Ref.~\cite{qprk-k7wn}; for completeness,
  a self-contained proof is given in Appendix~\ref{proof:single-string-sm}.
  \end{proof}

\begin{lemma}[Gaussian-cone closure]
\label{lem:cone-closure}
The following operations preserve $\FGC$:
\begin{enumerate}[label=(\roman*)]
\item positive finite sums, norm-convergent countable sums, and positive averages of integrable $\FGC$-valued operator families over finite measures;
\item ordered graded-CAR products of elements supported on pairwise disjoint even Majorana subspaces;
\item Gaussian congruences
  \[
    A\longmapsto GAG^\dagger,
  \]
where $G$ is an invertible fermionic Gaussian operator, in particular $G=e^{-Q/2}$ for Hermitian quadratic $Q$;
\item taking the reduced state on an ordered fermionic subsystem.
\end{enumerate}
\end{lemma}

\begin{proof}
The itemized proof is given in Appendix~\ref{app:cone-closure}.
\end{proof}

\section{Main results}

We first introduce the structure of a sparse fermionic Hamiltonian.
Let
\begin{align}
    H=\sum_{a\in[K]}h_a
\end{align}
be a sum of even local fermionic operators, where
$\|h_a\|_{\gamma,1}\le g_a,J_a:=\Supp(h_a), |J_a|\le R.$
 The interaction graph has vertex set $[K]$ and an edge $(a,b)$ whenever $J_a\cap J_b\neq\varnothing$. The Hamiltonian is $(R,d)$-sparse if this graph has maximum degree at most $d$
 and the support of each term is at most $R$.
All local terms containing a fixed Majorana index are pairwise adjacent. Consequently, at most $d+1$ local terms can contain one fixed Majorana index.
In Theorem~\ref{thm:main}, the Majorana support size bound is imposed only on the perturbing family.

\begin{thm}[Sparse close-to-quadratic convex-Gaussian theorem]\label{thm:main}
Let $R\ge 2$, $\mathcal E>0$, and $\epsilon>0$.  Consider the even fermionic Hamiltonian
\begin{align}
    H=H_0+V=\sum_{b\in\mathcal A_0}q_b+\sum_{a\in\mathcal A_1}h_a,
\end{align}
where $\mathcal A_0$ is the set of labels for quadratic terms, 
and $\mathcal A_1$ is the set of labels for nonquadratic terms.
Assume:
\begin{enumerate}[label=(\roman*)]
\item each $q_b$ is quadratic, Hermitian, and $\|q_b\|_{\gamma,1}\le \mathcal E$;
\item each $h_a$ is even, Hermitian, has Majorana support size at most $R$, and $\|h_a\|_{\gamma,1}\le \epsilon\mathcal E$;
\item the interaction graph  has maximum degree $\leq d$.
\end{enumerate}
If the inverse temperature satisfies
\begin{equation}\label{eq:main-window}
\beta\le \frac{1}{R(d+1)\mathcal E}\log\!\left(1+\frac{1}{72\epsilon}\right),
\end{equation}
then the Gibbs state $\rho_\beta(H) \in \operatorname{conv}(\FG)_n$.
\end{thm}
\begin{proofsketch}
The complete proof is given in Appendix~\ref{app:sparse-proof}. Here, we provide a sketch proof. Let us expand the
interaction-picture perturbation exactly in the Majorana basis. At each deletion,
the support-decoupling construction records the selected remaining label, samples
two independent connected propagator histories, and samples one of seven algebraic
branches. Let $\Omega_{\mathrm{term}}$ be the countable set of complete terminal
records and let $p$ be the resulting probability distribution. The unnormalized
Gibbs state then has the following exact representation.
\begin{equation}
  e^{-\beta H}
  =
  \mathbb E_{\omega\sim p}\!\left[
    e^{-\beta H_0/2}
    \prod_j\bigl(I+\alpha_{\omega,j}\gamma_{J_{\omega,j}}\bigr)
    e^{-\beta H_0/2}
  \right],
  \label{eq:terminal-representation}
\end{equation}
where every $\omega\in\Omega_{\mathrm{term}}$ records a full deletion history,
$
  |\alpha_{\omega,j}|\leq1
$
and the nonempty supports $J_{\omega,j}$ are pairwise disjoint.
Lemma~\ref{lem:single-string} implies that each terminal affine Majorana factor is in the fermionic Gaussian cone. Lemma~\ref{lem:cone-closure} then applies to the disjoint graded-CAR product, the Gaussian congruence by $e^{-\beta H_0/2}$, and the positive expectation. Thus $e^{-\beta H}\in\FGC$.
\end{proofsketch}

\begin{proposition}[Asymptotic tightness]\label{prop:tight}
There exists a family of $(R,d)$-sparse fermionic Hamiltonians with fixed $R=4$, $d=4$, and quadratic local scale $\mathcal E=1$, whose nonquadratic local scale is $\epsilon$, such that 
the Gibbs state is not convex-Gaussian whenever
\begin{align}
    \beta>\Omega\left(\log\left(\frac{1}{\epsilon}\right)\right).
\end{align}
\end{proposition}
\begin{proofsketch}
Let us consider 
a specific even  fermionic Hamiltonian
as follows
\begin{align}
  H_\epsilon=\sum_{j=1}^{4}i\gamma_{2j-1}\gamma_{2j}
 -\epsilon\gamma_1\gamma_3\gamma_5\gamma_7.  
\end{align}
Then we can show that the Gibbs state for $H_\epsilon$ is not convex-Gaussian whenever
$ \beta>\Omega\left(\log\left(\frac{1}{\epsilon}\right)\right).$
The complete proof
is given in Appendix~\ref{app:tightness}.
\end{proofsketch}

Theorem~\ref{thm:main}  and Ref.~\cite{qprk-k7wn} address complementary regimes. Ref.~\cite{qprk-k7wn} establishes a uniform, system-size-independent convex-Gaussian regime for bounded-degree local fermionic Hamiltonians. Theorem~\ref{thm:main} instead provides a perturbative bound showing that convex-Gaussianity extends to parametrically lower temperatures near the free-fermion limit. For fixed locality and degree, the inverse-temperature window scales as $\beta \mathcal{E} = \Theta(\log(1/\epsilon))$ as the nonquadratic interaction strength $\epsilon \to 0$.
Our result thus accesses regimes beyond what generic high-temperature expansions can capture.  Proposition~\ref{prop:tight} 
further shows that the logarithmic dependence on $\epsilon$ is
asymptotically optimal up to universal constants.

We now apply our results to an important class of fermionic Hamiltonians, namely the Fermi--Hubbard model~\cite{Scalapino2007,PhysRevX.5.041041,RevModPhys.80.885,Mazurenko2017}.

\emph{Fermi--Hubbard model.}
Consider a finite graph with vertex set $\Lambda$ ($|\Lambda|=L$), edge set $E$, and maximum degree $D$.
The Fermi--Hubbard Hamiltonian in the centered convention is given as follows
\begin{equation}\label{eq:hubbard}
H_{U,t}=tH_{\rm hop}+UV_{\rm int},
\end{equation}
where 
\begin{align}
  H_{\rm hop}&:=-\sum_{\{x,y\}\in E}\ \sum_{\sigma\in\{\uparrow,\downarrow\}}
  (c_{x\sigma}^\dagger c_{y\sigma}+c_{y\sigma}^\dagger c_{x\sigma}),\\ 
  V_{\rm int}&:=\frac14\sum_{x\in\Lambda}V_x,\quad  V_x:=(2n_{x\uparrow}-I)(2n_{x\downarrow}-I).
\end{align}

For fixed $D$, let us first consider the weak-coupling regime, i.e., $|U|/|t|\to0$.
When $U=0$, the Hamiltonian is quadratic, so the Gibbs state belongs to $\operatorname{conv}(\FG)_{2L}$ for arbitrarily low temperatures $T$.
In addition, we have the following result for the general weak-coupling regime.

\begin{proposition}[Fermi--Hubbard model in weak-coupling regime]
\label{prop:hubbard-weak}
Consider the Fermi--Hubbard Hamiltonian~\eqref{eq:hubbard} in the weak-coupling regime.
If the inverse temperature satisfies
\begin{align}
    \beta\le\frac{1}{4(2D+1)|t|}
 \log\!\left(1+\frac{|t|}{18|U|}\right),
 \label{eq:weak-window}
\end{align}
then the Gibbs state $\rho_{\beta}(H_{U,t})\in  \operatorname{conv}(\FG)_{2L} $. 
\end{proposition}
\begin{proof}
The result comes from Theorem ~\ref{thm:main} by choosing
$R=4$, $\mathcal E=|t|$, $d\leq 2D$, and
$\epsilon=\frac{|U|}{4|t|}$. The parameter verification is given in
Appendix~\ref{app:weak-hubbard}.

\end{proof}

As $U \to 0$, the above proposition establishes that the threshold on $\beta$ grows arbitrarily high, meaning the Gibbs state of the interacting fermions can remain convex-Gaussian even at near-zero temperatures.

We now turn to the Fermi--Hubbard model in the strong-coupling regime ($|t| \ll |U|$)~\cite{DUPUIS_2000,PhysRevA.105.033317}. Starting from the atomic limit $t = 0$, where the Hamiltonian contains only nonquadratic terms, we establish that the corresponding Gibbs state is a convex combination of Gaussian states for all inverse temperatures $\beta$.

\begin{proposition}[Atomic limit]
\label{prop:atomic}
Consider the Fermi--Hubbard model in \eqref{eq:hubbard}. If the hopping
$t=0$, then the Gibbs state $ \rho_\beta(H_{U,0})
  \in
  \operatorname{conv}(\FG)_{2L}$
  for any inverse temperature $\beta$.

\end{proposition}

\begin{proof}
Since $t=0$, the Fermi--Hubbard Hamiltonian in \eqref{eq:hubbard} 
will be simplified to $H_{U,0}
  =
  \frac U4\sum_{x\in\Lambda}V_x,$
where the operators $V_x$ have pairwise disjoint site supports. For each site,
\begin{align}
      e^{-\beta UV_x/4}
  =
  \cosh\!\left(\frac{\beta U}{4}\right)
  \left[
    I-
    \tanh\!\left(\frac{\beta U}{4}\right)V_x
  \right].
\end{align}
Since $V_x$ is a normalized even Majorana string up to a sign and
$ \left|\tanh\!\left(\beta U/4\right)\right|\leq1$,
Lemma~\ref{lem:single-string} implies that each site factor is in its supported fermionic Gaussian cone. The site supports are disjoint, so Lemma~\ref{lem:cone-closure} implies that their ordered graded-CAR product also belongs to the fermionic Gaussian cone. Hence, the Gibbs state $\rho_\beta(H_{U,0})$ is a convex combination of Gaussian states for any inverse temperature $\beta$.
\end{proof}

The strong-coupling regime differs conceptually from the weak-coupling setting of Theorem~\ref{thm:main}. In the latter, convex-Gaussianity is obtained by perturbing a quadratic, free-fermion reference Hamiltonian. 
Here, in contrast, the unperturbed Hamiltonian $UV_{\mathrm{int}}$ features quartic interactions. 
The argument of
Theorem~\ref{thm:main} cannot be applied simply by exchanging the roles of the
quadratic and interacting parts. Although Proposition~\ref{prop:atomic} shows that
the atomic Gibbs state is convex-Gaussian for every inverse
temperature, congruence by a general convex-Gaussian operator does
not in general preserve the fermionic Gaussian cone.
The strong-coupling result therefore relies on a different
mechanism.  The commuting on-site structure of
$V_{\mathrm{int}}$ implies a model-specific local
Gaussian-cone property: certain Majorana-string corrections can be combined with the atomic Gibbs factors while staying inside the fermionic Gaussian cone.  Combined with a site-based support-decoupling construction,
this allows the terminal Majorana factors to be grouped into
pairwise disjoint site blocks and treated locally.  Theorem~\ref{thm:hubbard-strong} shows
that this atomic Gaussian-cone structure is stable under
sufficiently weak hopping.  Thus convex-Gaussianity can persist
perturbatively around an interacting, nonquadratic reference
Hamiltonian, by a mechanism distinct from perturbation around the
free-fermion limit.

\begin{thm}[Strong-coupling Fermi--Hubbard theorem]
\label{thm:hubbard-strong}
Consider the Fermi--Hubbard model~\eqref{eq:hubbard} with $U\neq0$. If the inverse temperature $\beta$ satisfies 
\begin{equation}
  \frac{8D|t|}{|U|}
  \left(e^{\beta|U|/2}-1\right)
  \leq
  \frac1{72},
  \label{eq:strong-condition}
\end{equation}
then 
the Gibbs state 
$\rho_\beta(H_{U,t})
  \in
  \operatorname{conv}(\FG)_{2L}.$
In addition, for $t\neq0$, condition \eqref{eq:strong-condition} is equivalent to
\begin{equation}
  \beta
  \leq
  \frac2{|U|}
  \log\!\left(1+\frac{|U|}{576D|t|}\right).
  \label{eq:strong-window}
\end{equation}
\end{thm}
\begin{proofsketch}
The complete proof is given in Appendix~\ref{app:strong-proof}. 
The proof is similar to that of Theorem \ref{thm:main}.
One key point is to show that the unnormalized 
Gibbs state can be written as 
\begin{align}
   e^{-\beta H_{U,t}}=\mathbb{E}_{\omega}\left[e^{-\beta UV_{\rm int}/2}
 \prod_j(I+\alpha_{\omega,j} \gamma_{J_{\omega,j}})
 e^{-\beta UV_{\rm int}/2}\right],
\end{align}
which belongs to $\FGC$. This holds because
\begin{align}
\label{eq:key_prop}
   e^{-\beta UV_{\rm int}/2}
   (I+\alpha X)   e^{-\beta UV_{\rm int}/2}\in \FGC,
\end{align}
where $|\alpha|\leq 1$ and $X$ is any normalized Hermitian even Majorana
monomial. Note that equation \eqref{eq:key_prop} is not a consequence of the Gaussian congruences in Lemma~\ref{lem:cone-closure}, but rather stems from the following model-specific local Gaussian-cone statement of $V_{\rm int}$: 
if $\gamma_J$ has odd Majorana degree on each of two sites $x,y$, then for all
real $a_x,a_y,$ and $|r|\leq 1$,
\begin{align}
    e^{-a_xV_x-a_yV_y}+r\gamma_J\in \FGC_4.
\end{align}
The above equation is a generalization of Lemma \ref{lem:single-string} for the Fermi--Hubbard model.
\end{proofsketch}

Theorem~\ref{thm:hubbard-strong} complements the discussion of Proposition~\ref{prop:hubbard-weak}. The weak-coupling result expands around the quadratic hopping Hamiltonian and is useful when $|U|\ll|t|$, whereas the strong-coupling result expands around the atomic interaction and is useful when $|t|\ll|U|$.

We now consider a specific example of the two-dimensional Fermi-Hubbard model, for which 
$D=4$.  In Fig.~\ref{fig:hubbard-certified-windows}, 
the horizontal axis represents the coupling ratio $\epsilon=|U|/(4|t|)$. 
The solid curve denotes the equality boundary of Eq.~\eqref{eq:weak-window},
read along the left $\beta|t|$ axis; it behaves like $\log(1/\epsilon)$ for small $\epsilon$ and $1/\epsilon$ for large $\epsilon$. 
The dashed curve denotes the equality boundary of Eq.~\eqref{eq:strong-window}, read along the right 
$\beta|U|$ axis; it behaves like $\epsilon$ for small $\epsilon$ and $\log\epsilon$ for large $\epsilon$. Each curve bounds its corresponding dimensionless inverse temperature. For better figure display, we further compactify the coordinates $X$ and $Y$ via
\[
\begin{aligned}
  X&=\arctan\!\left(\frac14\log\frac{\epsilon}{\sqrt8}\right),\qquad \upsilon=200,\\
  Y_t&=\frac{36\upsilon\,\beta|t|}{1+36\upsilon\,\beta|t|},\qquad
  Y_U=\frac{\upsilon\,\beta|U|/2}{1+\upsilon\,\beta|U|/2}.
\end{aligned}
\]
Figure~\ref{fig:hubbard-certified-windows} highlights the regime where the Gibbs state of the Fermi–Hubbard model exhibits ``classical'' behavior, defined here as a convex combination of Gaussian states.

\begin{figure}[H]
\centering
\includegraphics{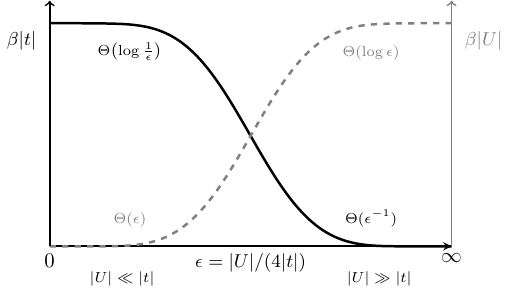}
\caption{The weak- and strong-coupling windows for the 2-dim Fermi--Hubbard model on a square lattice. The solid curve is read on the left axis and the dashed curve on the right axis. The asymptotic labels refer to the bounds before compactification.}
\label{fig:hubbard-certified-windows}
\end{figure}

\section{Conclusion}
In this work, we have investigated the convex-Gaussianity of fermionic Gibbs states in perturbation theory and delineated its precise parameter regimes. First, we present a general weak-coupling theorem showing that for a sparse Hamiltonian $H = H_0 + V$ with a quadratic reference $H_0$ and a non-quadratic perturbation $V$ of scale $\epsilon$, the Gibbs state $\rho_{\beta}$ decomposes into a convex combination of Gaussian states for $\beta \le O(\log(1/\epsilon))$, a bound we establish as asymptotically tight.
Applying this framework to the Fermi--Hubbard model covers the weak-coupling (small-$\vert{}U\vert{}$) regime. Furthermore, for the strong-coupling (small-$\vert{}t\vert{}$) regime, we demonstrate that convex-Gaussianity persists up to $\beta \le O\big(\vert{}U\vert{}^{-1} \log(\vert{}U\vert{} / (D\vert{}t\vert{}))\big)$, revealing a mechanism for convex-Gaussianity distinct from that of the weak-coupling setting.

Our results identify two perturbative mechanisms for convex-Gaussianity: one around a quadratic free-fermion Hamiltonian and another around the interacting atomic limit of the Fermi–Hubbard model.
However, several key challenges remain open, such as efficient sampling of the Gaussian decompositions, optimized locality dependence, and the tightness of the strong-coupling result in Theorem~\ref{thm:hubbard-strong}.

\subsection*{Acknowledgements and AI disclosure}
The authors gratefully acknowledge the assistance of OpenAI's Codex (GPT-5.6 Sol) in carrying out part of the derivations and
implementing the Lean 4 formalization.
Y.R. thanks  Jiaqing Jiang, Yu Tong, Ruihua Fan, Roland Bauerschmidt, and Joel Feldman for helpful discussions.
K. B. thanks Arthur Jaffe and Seth Lloyd for helpful discussions.
Y.R. was supported by the U.S.
Department of Energy, Office of Science, National Quantum Information
Science Research Centers, Quantum Systems Accelerator, under Grant
number DOE DE-SC0012704.
K. B. acknowledges the support from NSF grant DMS-2606749.

\bibliography{published-reference}
\clearpage
\onecolumngrid
\setcounter{secnumdepth}{3}
\appendix
\setcounter{equation}{0}
\renewcommand{\theequation}{S\arabic{equation}}

\makeatletter
\@removefromreset{equation}{section}
\makeatother

\section{Proof of the sparse close-to-quadratic theorem}
\label{app:sparse-proof}
\subsection{Convex-Gaussian and sparse-interaction preliminaries}
\subsubsection{Majorana strings and the Majorana coefficient norm}
Let $[2n]:=\{1,\ldots,2n\}$ be the set of Majorana indices. The Majorana operators satisfy
\begin{equation}
\{\gamma_x,\gamma_y\}=2\delta_{xy}I,\qquad x,y\in[2n].
\end{equation}
For an even set $J=\{j_1<\cdots<j_{2r}\}\subseteq[2n]$, define
\begin{equation}
\gamma_J:=i^r\gamma_{j_1}\cdots\gamma_{j_{2r}}.
\end{equation}
Then $\gamma_J=\gamma_J^\dagger$ and $\gamma_J^2=I$. Every even operator has a unique expansion
\begin{equation}
O=\sum_{J:\,|J|\ \mathrm{even}}o_J\gamma_J.
\end{equation}
We use the Majorana coefficient norm
\begin{equation}
\|O\|_{\gamma,1}:=\sum_J|o_J|.
\end{equation}

\begin{lemma}[Basic Majorana coefficient norm bounds]
\label{lem:coefficient-norm-bounds-sm}
For even operators $A,B$,
\[
\|AB\|_{\gamma,1}\le \|A\|_{\gamma,1}\|B\|_{\gamma,1},\qquad
\|[A,B]\|_{\gamma,1}\le 2\|A\|_{\gamma,1}\|B\|_{\gamma,1}.
\]
Here $[A,B]:=AB-BA$ denotes the commutator.
If $A$ and $B$ have disjoint Majorana supports, then $[A,B]=0$.
\end{lemma}
\begin{proof}
The product of two Majorana strings is a phase times one Majorana string. Expanding $AB$ and applying the triangle inequality gives the first inequality; the commutator bound follows from $AB-BA$. Even strings on disjoint Majorana supports commute, and linearity gives the last claim.
\end{proof}

\subsubsection{Fermionic Gaussian cone}
\label{app:cone-closure}
Here a normalized fermionic Gaussian state may be pure or mixed: it is a trace-one positive operator in the norm-closed quadratic-Gaussian class. Let $\FGC$ denote the cone of finite nonnegative combinations of normalized fermionic Gaussian states; its trace-one slice is the set of convex-Gaussian states.
We omit the mode subscript when the number of modes is clear.

\begin{lemma}[Gaussian-cone closure]
The following operations preserve $\FGC$:
\begin{enumerate}[label=(\roman*)]
\item positive finite sums, norm-convergent countable sums, and positive averages of integrable $\FGC$-valued operator families over finite measures;
\item ordered graded-CAR products of elements supported on pairwise disjoint even Majorana subspaces;
\item Gaussian congruence $A\mapsto GAG^\dagger$, where $G$ is an invertible fermionic Gaussian operator, in particular $G=e^{-Q/2}$ for Hermitian quadratic $Q$.
\item 
taking the reduced state on an ordered fermionic
subsystem.
\end{enumerate}
\end{lemma}
\begin{proof}
\begin{enumerate}[label=(\roman*)]
\item The normalized fermionic Gaussian base is compact in finite dimension,
and so is its convex hull. Trace normalization therefore shows that its finite
nonnegative cone is closed: a positive-trace limit is its trace times a limit
point of the compact normalized convex hull, while a positive zero-trace limit
is zero. Norm-convergent positive sums and positive averages remain in this
closed cone.
\item An ordered CAR embedding $\iota$ satisfies
\[
 \iota(AB)=\iota(A)\iota(B),\qquad
 \iota(A^\dagger)=\iota(A)^\dagger,\qquad
 \iota(I)=I.
\]
Let $\iota_k$ be the canonical ordered CAR embeddings of pairwise
disjoint even Majorana subspaces. For $A_k\in\FGC$ on the corresponding
subspaces,
\[
 [\iota_k(A_k),\iota_\ell(A_\ell)]=0\quad(k\ne\ell),
 \qquad
 \prod_k\iota_k(A_k)\in\FGC.
\]
The ordered embeddings contain the required parity strings, and the displayed
graded product preserves fermionic Gaussianity.
\item Congruence by a quadratic-Gaussian group element preserves the positive
fermionic Gaussian class, its norm closure, and finite nonnegative
combinations.
\item
Let $\iota$ be a canonical ordered CAR embedding from $m$-mode operators into
$n$-mode operators.  For an $n$-mode operator $\rho$, define its restriction
$\rho_\iota$ as the unique $m$-mode operator satisfying
\[
  \operatorname{Tr}_m(A\rho_\iota)
  =\operatorname{Tr}_n(\iota(A)\rho)
  \qquad\text{for every $m$-mode operator $A$}.
\]
The Majorana expansion constructs $\rho_\iota$, proves uniqueness, and shows
that restriction is linear.  Taking $A=I$ shows that trace is preserved.  For
a normalized fermionic Gaussian state $\rho$ and every rank-one positive
operator $B^\dagger B$,
\[
  \operatorname{Tr}_m(B^\dagger B\rho_\iota)
  =\operatorname{Tr}_n\!\bigl(\iota(B)\rho\iota(B)^\dagger\bigr)
  \ge0,
\]
so positivity is preserved.  Taking $A$ to be a Majorana basis operator
preserves every embedded Majorana moment, including the two-point
correlations.  Because the embedding is ordered, Wick's theorem is preserved.
Thus restriction preserves normalized fermionic Gaussian states.  Finally, if
$C=\sum_i p_i\rho_i\in\FGC$ with $p_i\ge0$, then linearity gives
\[
  C_\iota=\sum_i p_i(\rho_i)_\iota\in\FGC.
\]
\end{enumerate}
\end{proof}

\phantomsection
\label{proof:single-string-sm}
\begin{proof}[Proof of Lemma~\ref{lem:single-string}]
It is enough to prove the result for $\alpha=\pm1$, because
\[
I+\alpha \gamma_J
=
\frac{1+\alpha}{2}(I+\gamma_J)
+
\frac{1-\alpha}{2}(I-\gamma_J).
\]
For $J=\varnothing$, the two endpoints are $2I$ and $0$.

Suppose $J\neq\varnothing$. We use induction on $|J|/2$. If $|J|=2$, then $I\pm \gamma_J$ is proportional to a fermionic Gaussian state. For the induction step, write $J=J_1\sqcup J_2$, where $J_2$ consists of the final two ordered Majorana indices. The supports of $\gamma_{J_1}$ and $\gamma_{J_2}$ are disjoint, both operators are even, and $\gamma_{J_1}^2=\gamma_{J_2}^2=I$. Hence they commute and
\[
I+\gamma_{J_1}\gamma_{J_2}=\frac12(I+\gamma_{J_1})(I+\gamma_{J_2})+\frac12(I-\gamma_{J_1})(I-\gamma_{J_2}).
\]
Similarly,
\[
I-\gamma_{J_1}\gamma_{J_2}=\frac12(I+\gamma_{J_1})(I-\gamma_{J_2})+\frac12(I-\gamma_{J_1})(I+\gamma_{J_2}).
\]
The induction hypothesis and closure of the fermionic Gaussian cone under ordered graded-CAR products on disjoint even supports prove the claim.
\end{proof}

\subsubsection{Sparse interaction geometry}
\begin{Def}[Sparse fermionic Hamiltonian]
Let
\[
H=\sum_{a\in[K]}h_a
\]
be a sum of even local fermionic operators, where
\[
\|h_a\|_{\gamma,1}\le g_a,\qquad \Supp(h_a)\subseteq J_a,\qquad |J_a|\le R,
\]
where $J_a$ is a chosen Majorana support set. The interaction graph has vertex set $[K]$ and an edge $(a,b)$ whenever $J_a\cap J_b\neq\varnothing$. The Hamiltonian is $(R,d)$-sparse if this graph has maximum degree at most $d$.
\end{Def}
All local terms containing a fixed Majorana index are pairwise adjacent. Consequently, at most $d+1$ local terms can contain one fixed Majorana index.

\subsection{Quadratic interaction-picture locality and local perturbation-strength bounds}
For the general argument below, write $H=H_0+V$, where $H_0$ is a quadratic
Hamiltonian and $V=\sum_a h_a$ is a finite sum of even Hermitian perturbation
terms.  For $\beta\geq0$, write
$\rho_\beta(H):=e^{-\beta H}/\Tr{e^{-\beta H}}$.

\subsubsection{One-particle propagation}
With $H_0$ as above, define the one-particle generator $\mathcal K$ by
\begin{equation}
[H_0,\gamma_x]=\sum_{y\in[2n]}\mathcal K_{xy}\gamma_y.
\end{equation}
Define
\begin{equation}
\|B\|_0:=\max\left\{\sup_x\sum_y|B_{xy}|,\
\ \sup_y\sum_x|B_{xy}|\right\},
\qquad \kappa_0:=\|\mathcal K\|_0.
\end{equation}
Note that the maximal row-sum and column-sum norm in the
Pauli basis has been used in the study of quantum machine learning~\cite{PhysRevA.105.062431,Bu_2023}.

\begin{lemma}[One-particle propagation]\label{lem:one-particle}
Let $M(t)=e^{t\mathcal K}$. Then
\begin{equation}
e^{tH_0}\gamma_xe^{-tH_0}=\sum_yM_{xy}(t)\gamma_y,
\end{equation}
and
\begin{equation}
\|M(t)\|_0\le e^{|t|\kappa_0}.
\end{equation}
\end{lemma}
\begin{proof}
Differentiating the evolved Majorana gives the linear ODE
$M'(t)=\mathcal K M(t)$ with $M(0)=I$, so $M(t)=e^{t\mathcal K}$.
The maximum row-sum and column-sum norm is submultiplicative, and expansion of the
exponential gives the bound.
\end{proof}

\subsubsection{Propagation of the local perturbation strength}
Write
\[
h_a=\sum_Jz_{a,J}\gamma_J,\qquad
h_a(t):=e^{tH_0}h_ae^{-tH_0}=\sum_Kz_{a,K}(t)\gamma_K.
\]
Define the initial and evolved local perturbation strengths
\begin{equation}
\label{eq:def_s0}
s_0:=\sup_{x\in[2n]}\sum_a\sum_{J\ni x}|z_{a,J}|,
\end{equation}
\begin{equation}
s(t):=\sup_y\sum_a\sum_{K\ni y}|z_{a,K}(t)|.
\end{equation}

\begin{thm}[Local-strength propagation]
\label{thm:local_str_prop}
Suppose every Majorana string $\gamma_J$ appearing in a perturbation term satisfies $|J|\le R$. Then
\begin{equation}
s(t)\le s_0e^{R\kappa_0|t|}.
\end{equation}
\end{thm}
\begin{proof}
Fix an output Majorana index $y\in[2n]$. We first estimate the contribution of a single input Majorana string
$\gamma_J$, where $J=\{j_1,\ldots,j_q\},q\le R$. 
By the one-particle propagation formula,
\[
e^{tH_0}\gamma_{j_r}e^{-tH_0}=\sum_{x_r\in[2n]}M_{j_rx_r}(t)\gamma_{x_r}.
\]
Hence, 
\begin{equation}\label{eq:local-strength-tuple-expansion}
e^{tH_0}\gamma_Je^{-tH_0}
=\zeta_J\prod^q_{r=1} \left(e^{tH_0}\gamma_{j_r}e^{-tH_0}\right)
=\zeta_J\sum_{x_1,\ldots,x_q}
\left(\prod_{r=1}^q M_{j_rx_r}(t)\right)
\gamma_{x_1}\cdots\gamma_{x_q},
\end{equation}
where $\zeta_J$ is a phase with  $|\zeta_J|=1$,

The tuple $(x_1,\ldots,x_q)$ may contain repeated Majorana indices. After applying the canonical anticommutation relations, its product reduces, up to phase, to a normalized Majorana string $\gamma_K$. Write $z_K^{(J)}(t)$ for the coefficient of $\gamma_K$ in this expansion. If the reduced support $K$ contains $y$, then $y$ must occur at least once among the tuple entries. Therefore, discarding cancellations and overcounting repeated occurrences can only increase the total coefficient mass, and the total mass of output strings containing $y$ that originate from $\gamma_J$ satisfies
\begin{align}\label{eq:single-input-incidence}
\sum_{K\ni y}|z^{(J)}_K(t)|
\le \sum_{r=1}^q
\sum_{\substack{x_1,\ldots,x_q\\x_r=y}}
\prod_{\ell=1}^q|M_{j_\ell x_\ell}(t)|
=\sum_{r=1}^q |M_{j_ry}(t)|
\prod_{\ell\ne r}\left(\sum_u|M_{j_\ell u}(t)|\right).
\end{align}
The row-sum bound in Lemma~\ref{lem:one-particle} gives
\begin{equation}\label{eq:row-mass-M}
\sup_x\sum_u|M_{xu}(t)|\le e^{\kappa_0|t|}.
\end{equation}
Consequently,
\begin{equation}\label{eq:single-input-incidence-2}
\sum_{K\ni y}|z^{(J)}_K(t)|
\le e^{(q-1)\kappa_0|t|}\sum_{r=1}^q|M_{j_ry}(t)|.
\end{equation}
Multiply by $|z_{a,J}|$, sum over $a$ and $J$, and use $q\le R$. Regrouping the distinguished input Majorana $j_r$ as $x$ gives
\begin{align}
\sum_a\sum_{K\ni y}|z_{a,K}(t)|
&\le e^{(R-1)\kappa_0|t|}
\sum_a\sum_J|z_{a,J}|\sum_{x\in J}|M_{xy}(t)|\nonumber\\
&=e^{(R-1)\kappa_0|t|}
\sum_x|M_{xy}(t)|
\left(\sum_a\sum_{J\ni x}|z_{a,J}|\right)\nonumber\\
&\le s_0e^{(R-1)\kappa_0|t|}\sum_x|M_{xy}(t)|,
\end{align}
where the last inequality comes from the definition of $s_0$ in \eqref{eq:def_s0}.
Because the definition of $\|M(t)\|_0$ controls column sums as well as row sums,
\[
\sup_y\sum_x|M_{xy}(t)|\le e^{\kappa_0|t|}.
\]
Taking the supremum over $y$ proves
\[
s(t)\le s_0e^{R\kappa_0|t|}.
\]

\end{proof}

\begin{cor}[Sparse bounds]
If the interaction graph on all local terms has maximum degree $d$ and each perturbation term satisfies $\|h_a\|_{\gamma,1}\le g$, then
\begin{equation}
s_0\le(d+1)g.
\end{equation}
If each quadratic local term $q_b$ satisfies $\|q_b\|_{\gamma,1}\le \mathcal E$, then
\begin{equation}
\kappa_0\le2(d+1)\mathcal E.
\end{equation}
\end{cor}
\begin{proof}
First, by the definition of the interaction graph, there are
at most $d+1$ terms containing a fixed Majorana index. One perturbation term contributes at most its full Majorana coefficient norm to $s_0$. In addition, for a quadratic term, the commutator with one Majorana operator introduces coefficients of total magnitude at most twice the quadratic Majorana coefficient norm. Summing over incident quadratic terms proves the second bound.
\end{proof}

\subsubsection{Majorana-basis expansion and integrated overlap strength}

We can always replace $H$ by $H-\text{constant} I$, which does not change the Gibbs states. In addition, 
removing these scalar coefficients does not increase the Majorana coefficient norms or enlarge the chosen Majorana support sets, and it leaves the local strengths unchanged. We therefore relabel the scalar-free Hamiltonian and perturbation terms as $H$ and $h_a$ and retain $z_{a,J}$ and $z_{a,K}(t)$ for their static and evolved coefficients. The nonempty coefficients are unchanged, while $z_{a,\varnothing}=z_{a,\varnothing}(t)=0$. The case $\beta=0$ is immediate because the Gibbs state is the normalized identity; hence below we assume $\beta>0$.
Fix
$\tau_\beta:=\frac{\beta}{2}$.
For every perturbation term,
\begin{equation}
h_a(t)=\sum_{K:\,0<|K|\le R}z_{a,K}(t)\gamma_K.
\end{equation}
Define the finite virtual-label set
\begin{equation}
\mathcal L:=\{\lambda=(a,K):z_{a,K}\not\equiv0\text{ on }[-\tau_\beta,\tau_\beta]\}.
\end{equation}
For any virtual label $\lambda=(a,K)$, we have
\[
P_\lambda:=\gamma_K,\qquad z_\lambda(t):=z_{a,K}(t),\qquad \Supp(\lambda):=K.
\]
Then
\begin{equation}
e^{tH_0}Ve^{-tH_0}=\sum_{\lambda\in\mathcal L}z_\lambda(t)P_\lambda.
\end{equation}

\begin{lemma}[Time reflection of the virtual coefficients]
\label{lem:tim_ref}
For every virtual label $\lambda$,
\begin{equation}
z_\lambda(-t)=\overline{z_\lambda(t)}.
\end{equation}
Equivalently,
\[
\left(e^{tH_0}Ve^{-tH_0}\right)^\dagger=e^{-tH_0}Ve^{tH_0}.
\]
\end{lemma}
\begin{proof}
Since $H_0$ and every $h_a$ are Hermitian,
\[
h_a(t)^\dagger=\left(e^{tH_0}h_ae^{-tH_0}\right)^\dagger=e^{-tH_0}h_ae^{tH_0}=h_a(-t).
\]
The basis operators $\gamma_K$ are Hermitian and linearly independent. Comparing coefficients in
\[
h_a(t)^\dagger=\sum_K\overline{z_{a,K}(t)}\gamma_K
\quad\text{and}\quad
h_a(-t)=\sum_Kz_{a,K}(-t)\gamma_K
\]
gives the result.
\end{proof}

\begin{Def}[Virtual activities and integrated overlap strength]
Define
\begin{equation}
a_\lambda:=2\int_0^{\tau_\beta}|z_\lambda(t)|\,dt=\int_{-\tau_\beta}^{\tau_\beta}|z_\lambda(t)|\,dt.
\end{equation}
The support of every virtual label is nonempty, and continuity, time reflection, and $\tau_\beta>0$ imply $a_\lambda>0$ for every $\lambda\in\mathcal L$.
Two labels overlap, written $\lambda\sim\nu$, if their Majorana supports
intersect. Define the integrated overlap strength
\begin{equation}
\Delta:=\sup_{\lambda\in\mathcal L}\sum_{\nu:\,\nu\sim\lambda}a_\nu,
\end{equation}
where the sum includes $\nu=\lambda$.
\end{Def}

\begin{lemma}[System-size-independent bound on the integrated overlap strength]
The integrated overlap strength obeys
\begin{equation}
\Delta\le2R\int_0^{\tau_\beta}s(t)\,dt\le2Rs_0\int_0^{\tau_\beta}e^{R\kappa_0t}\,dt.
\end{equation}
For $\kappa_0>0$,
\begin{equation}
\Delta\le\frac{2s_0}{\kappa_0}\left(e^{R\kappa_0\tau_\beta}-1\right),
\end{equation}
and for $\kappa_0=0$, the corresponding bound is $\Delta\le2R\tau_\beta s_0$.
\end{lemma}
\begin{proof}
For fixed $\lambda$,
\[
\sum_{\nu:\nu\sim\lambda}a_\nu
\le2\int_0^{\tau_\beta}\sum_{x\in\Supp(\lambda)}\sum_{\nu:\,x\in\Supp(\nu)}|z_\nu(t)|\,dt
\le2R\int_0^{\tau_\beta}s(t)\,dt.
\]
By the local-strength propagation in Theorem \ref{thm:local_str_prop}, we get the results.
\end{proof}

\subsection{Deletion propagator and rooted overlap histories}
For $S\subseteq\mathcal L$, define
\begin{equation}
V_S(t):=\sum_{\lambda\in S}z_\lambda(t)P_\lambda
\end{equation}
and let $U_S(t_2,t_1)$ solve
\begin{equation}
\partial_{t_2}U_S(t_2,t_1)=-V_S(t_2)U_S(t_2,t_1),\qquad U_S(t_1,t_1)=I.
\end{equation}
Based on the definition, for $S=\mathcal L$, 
$U_{\mathcal L}(t_2,t_1)=e^{t_2H_0}e^{-(t_2-t_1)H}e^{-t_1H_0}$.

By the time-reflection identity in Lemma~\ref{lem:tim_ref},
\begin{equation}
V_S(t)^\dagger=V_S(-t),
\end{equation}
which implies
\begin{equation}
U_S(t_2,t_1)^\dagger=U_S(-t_1,-t_2).
\end{equation}
Choose $\lambda_*\in S$, put $\widehat S=S\setminus\{\lambda_*\}$, and define
\begin{equation}
\mathcal R_{\lambda_*,S}:=U_{\widehat S}(0,\tau_\beta)U_S(\tau_\beta,0).
\end{equation}

\subsubsection{Propagator calculus and the deletion equation}
We first record the two-time identities that will be used in the proof. Uniqueness of solutions gives the composition law
\begin{equation}
U_S(t_3,t_2)U_S(t_2,t_1)=U_S(t_3,t_1)
\end{equation}
and hence $U_S(t_2,t_1)^{-1}=U_S(t_1,t_2)$. In addition to the defining derivative in the first time variable, one has
\begin{equation}
\partial_{t_1}U_S(t_2,t_1)=U_S(t_2,t_1)V_S(t_1).
\end{equation}
Indeed, from
\[
U_S(t_2,t_1+h)U_S(t_1+h,t_1)=U_S(t_2,t_1)
\]
and
\[
U_S(t_1+h,t_1)=I-hV_S(t_1)+O(h^2)
\]
one obtains
\[
U_S(t_2,t_1+h)=U_S(t_2,t_1)+hU_S(t_2,t_1)V_S(t_1)+O(h^2).
\]
In particular,
\begin{equation}
\frac{d}{dt}U_{\widehat S}(0,t)=U_{\widehat S}(0,t)V_{\widehat S}(t),\qquad
\frac{d}{dt}U_S(t,0)=-V_S(t)U_S(t,0).
\end{equation}

Define the deletion propagator up to time $t$ by
\begin{equation}
G(t):=U_{\widehat S}(0,t)U_S(t,0),\qquad G(0)=I.
\end{equation}
Differentiating both factors gives
\begin{align}
\frac{d G(t)}{dt}
&=U_{\widehat S}(0,t)V_{\widehat S}(t)U_S(t,0)
-U_{\widehat S}(0,t)V_S(t)U_S(t,0)\nonumber\\
&=U_{\widehat S}(0,t)\bigl(V_{\widehat S}(t)-V_S(t)\bigr)U_S(t,0)\nonumber\\
&=-z_{\lambda_*}(t)U_{\widehat S}(0,t)P_{\lambda_*}U_S(t,0).
\end{align}
Since $G(t)=U_{\widehat S}(0,t)U_S(t,0)$, left multiplication by $U_{\widehat S}(t,0)$ yields
\begin{equation}
U_S(t,0)=U_{\widehat S}(t,0)G(t).
\end{equation}
Thus the exact deletion equation is
\begin{equation}\label{eq:deletion-propagation}
\frac{d G(t)}{dt}=-A_*(t)G(t),\qquad
A_*(t):=z_{\lambda_*}(t)B_*(t),
\end{equation}
where
\begin{equation}
B_*(t):=U_{\widehat S}(0,t)P_{\lambda_*}U_{\widehat S}(t,0).
\end{equation}
The identity $G(\tau_\beta)=\mathcal R_{\lambda_*,S}$ shows that it suffices to expand the solution of \eqref{eq:deletion-propagation}.

\subsubsection{Rooted overlap histories and the two-stage Dyson expansion}
A rooted overlap history is a finite list $\boldsymbol\lambda=(\lambda_1,\ldots,\lambda_t)$ with $\lambda_1=\lambda_*$. For every $j>1$, either $\lambda_j=\lambda_*$ or $\lambda_j\sim\lambda_i$ for some $i<j$. This is the connectivity forced by a nonzero nested commutator: disjoint even Majorana strings commute.

\begin{thm}[Deletion propagator from the Majorana-basis expansion]\label{thm:connected-propagator}
Assume $\Delta>0$. Put
\begin{equation}
\theta:=\frac{a_{\lambda_*}}{\Delta}\in(0,1].
\end{equation}
There is a probability distribution over triples $(b,Y,\boldsymbol\lambda)$ such that
\begin{equation}
\mathbb E[I+bY]=\mathcal R_{\lambda_*,S},
\end{equation}
where $b\ge0$, $Y$ is the history operator, namely an even Majorana string up to a unit-modulus phase, and, whenever $b\ne0$, the rooted overlap history $\boldsymbol\lambda=(\lambda_1,\ldots,\lambda_t)$ has length $t\ge1$ and satisfies:
\begin{enumerate}[label=(\roman*)]
\item $\boldsymbol\lambda$ is a rooted overlap history;
\item $Y=\zeta P_{\lambda_1}\cdots P_{\lambda_t}$ for some $\zeta\in\mathbb C$ with $|\zeta|=1$; in particular the phase is not assumed to be a fourth root of unity;
\item
\begin{equation}
b\le\theta(3\Delta)^t.
\end{equation}
\end{enumerate}
\end{thm}
\begin{proof}
We expand the deletion equation in two stages. The outer Dyson expansion counts explicit selected-label occurrences $\lambda_*$, while the inner expansion expresses the dressing of each selected-label occurrence through nested commutators with labels in $\widehat S$.

\paragraph{1) outer Dyson expansion and root Dyson blocks.}
The Volterra form of the deletion equation~\eqref{eq:deletion-propagation} is
\begin{align}
    G(t)=I-\int_0^t A_*(s)G(s)\,ds.
\end{align}
Iterating gives the absolutely convergent finite-dimensional Dyson series
\begin{equation}\label{eq:outer-dyson-v4}
G(\tau_\beta)=I+\sum_{m\ge1}(-1)^m
\int_{0\le s_m\le\cdots\le s_1\le \tau_\beta}
A_*(s_1)\cdots A_*(s_m)\,d\mathbf s,
\end{equation}
where $\mathbf s=(s_1, s_2,...,s_m)$.
Since $A_*(s)=z_{\lambda_*}(s)B_*(s)$, each factor in \eqref{eq:outer-dyson-v4} contains one selected-label occurrence $\lambda_*$. We call one such outer factor together with its inner nested-commutator expansion a \emph{root Dyson block}. Thus the outer order $m$ is exactly the number of root Dyson blocks.

\paragraph{2) inner Dyson expansion of the dressed selected-label operator.}
For a fixed operator $B$, differentiating its conjugation and using the two endpoint derivatives gives
\begin{align*}
\frac{d}{d\tau}\left(U_{\widehat S}(0,\tau)B\,U_{\widehat S}(\tau,0)\right)
=U_{\widehat S}(0,\tau)V_{\widehat S}(\tau)B\,U_{\widehat S}(\tau,0)
-U_{\widehat S}(0,\tau)BV_{\widehat S}(\tau)U_{\widehat S}(\tau,0)=U_{\widehat S}(0,\tau)[V_{\widehat S}(\tau),B]U_{\widehat S}(\tau,0).
\end{align*}
Integrating once and iterating the same identity gives, for $0\le s\le \tau_\beta$,
\begin{equation}\label{eq:inner-dyson-v4}
B_*(s)=P_{\lambda_*}+\sum_{r\ge1}\int_{\Sigma_r(s)}
\ad_{V_{\widehat S}(\tau_r)}\cdots\ad_{V_{\widehat S}(\tau_1)}(P_{\lambda_*})\,d\boldsymbol\tau,
\end{equation}
where $\boldsymbol\tau=(\tau_1, \tau_2,..., \tau_r)$, and 
\[
\Sigma_r(s):=\{0\le\tau_r\le\cdots\le\tau_1\le s\}.
\]
Expanding $V_{\widehat S}(\tau)=\sum_{\nu\in\widehat S}z_\nu(\tau)P_\nu$
in \eqref{eq:inner-dyson-v4} yields
\begin{align}\label{eq:label-resolved-inner-v4}
B_*(s)
=P_{\lambda_*}+\sum_{r\ge1}\sum_{\nu_1,\ldots,\nu_r\in\widehat S}
\int_{\Sigma_r(s)}\left(\prod_{j=1}^rz_{\nu_j}(\tau_j)\right)
\ad_{P_{\nu_r}}\cdots\ad_{P_{\nu_1}}(P_{\lambda_*})\,d\boldsymbol\tau.
\end{align}

\paragraph{3) algebraic form and connectivity of one root Dyson block.}
Every $P_\nu$ is a normalized Hermitian even Majorana string. Two such strings either commute or anticommute; in the latter case their commutator is twice their product up to phase. Thus, whenever a nested commutator in \eqref{eq:label-resolved-inner-v4} is nonzero,
\begin{equation}\label{eq:nested-string-form-v4}
\ad_{P_{\nu_r}}\cdots\ad_{P_{\nu_1}}(P_{\lambda_*})
=2^r\zeta(\boldsymbol\nu)P_{\nu_r}\cdots P_{\nu_1}P_{\lambda_*},
\qquad |\zeta(\boldsymbol\nu)|=1.
\end{equation}
At the $j$th nested commutator, if $P_{\nu_j}$ had support disjoint from $P_{\lambda_*}$ and from every earlier $P_{\nu_i}$, then it would commute with their even Majorana string and the $j$th commutator would vanish. Hence every nonzero block sequence
\[
(\lambda_*,\nu_1,\ldots,\nu_r)
\]
is a rooted overlap history. Support overlap is used only as a necessary condition for a nonzero commutator; it need not be sufficient.

\paragraph{4) exact full Dyson-history expansion and overlap graph.}
We now combine the two Dyson expansions into a single explicit history formula. Fix an outer order $m\ge1$ and a vector
\[
\mathbf r=(r_1,\ldots,r_m)\in\mathbb N_0^m.
\]
For block $i$ choose labels $\nu_{i,1},\ldots,\nu_{i,r_i}\in\widehat S$ and define
\begin{equation}\label{eq:block-operator-Ci}
\mathcal C_i:=\ad_{P_{\nu_{i,r_i}}}\cdots\ad_{P_{\nu_{i,1}}}(P_{\lambda_*}),
\end{equation}
with the convention $\mathcal C_i=P_{\lambda_*}$ if $r_i=0$. Define the joint time domain
\begin{align}\label{eq:history-time-domain}
\mathfrak D_{m,\mathbf r}(\tau_\beta):=\Bigl\{
0\le s_m\le\cdots\le s_1\le \tau_\beta,\quad
0\le\tau_{i,r_i}\le\cdots\le\tau_{i,1}\le s_i,
\quad 1\le i\le m
\Bigr\}.
\end{align}
Substitution of \eqref{eq:label-resolved-inner-v4} into \eqref{eq:outer-dyson-v4} gives the exact expansion
\begin{align}\label{eq:full-history-expansion}
G(\tau_\beta)=I
&+\sum_{m\ge1}\sum_{r_1,\ldots,r_m\ge0}
\sum_{\substack{\nu_{i,j}\in\widehat S\\1\le i\le m,\ 1\le j\le r_i}}
(-1)^m
\int_{\mathfrak D_{m,\mathbf r}(\tau_\beta)}
\left(\prod_{i=1}^m z_{\lambda_*}(s_i)\right)
\left(\prod_{i=1}^m\prod_{j=1}^{r_i}z_{\nu_{i,j}}(\tau_{i,j})\right)
\mathcal C_1\cdots\mathcal C_m\,
 d\boldsymbol\tau\,d\mathbf s.
\end{align}
Since $G(\tau_\beta)=\mathcal R_{\lambda_*,S}$, \eqref{eq:full-history-expansion} is also the exact Dyson-history expansion of the deletion propagator.

A \emph{full Dyson history} is the data
\[
h=(m,\mathbf r,\mathbf s,\boldsymbol\tau,\boldsymbol\nu)
\]
appearing in one integrand of \eqref{eq:full-history-expansion}. Its total length is
\begin{equation}\label{eq:history-length-v4}
t(h):=m+\sum_{i=1}^m r_i.
\end{equation}
Introduce occurrence labels
\[
\lambda_{i,0}:=\lambda_*,\qquad \lambda_{i,j}:=\nu_{i,j}\quad(1\le j\le r_i).
\]
The ordered label list in the theorem is the concatenation of the blocks,
\begin{align}\label{eq:history-label-list-v4}
\boldsymbol\lambda(h)=(&\lambda_*,\nu_{1,1},\ldots,\nu_{1,r_1},
\lambda_*,\nu_{2,1},\ldots,\nu_{2,r_2},\ldots).
\end{align}

The \emph{overlap graph} $\Gamma(h)$ has occurrence vertex set
$\{(i,j):1\le i\le m,\ 0\le j\le r_i\}$ and an edge between $(i,j)$ and
$(k,\ell)$ precisely when
\begin{equation}\label{eq:history-overlap-edge}
\Supp(\lambda_{i,j})\cap\Supp(\lambda_{k,\ell})\ne\varnothing,
\end{equation}
i.e. $\lambda_{i,j}\sim\lambda_{k,\ell}$. If the block operator $\mathcal C_i$
is nonzero, then every non-root occurrence must overlap some earlier occurrence,
so the subgraph of $\Gamma(h)$ induced by $(i,0),\ldots,(i,r_i)$ is connected
and rooted at $(i,0)$. Since all root occurrences carry the same label
$\lambda_*$, they overlap one another; therefore every full Dyson history with all
block operators nonzero has connected $\Gamma(h)$.

Whenever all blocks are nonzero, \eqref{eq:nested-string-form-v4} shows that the operator contribution is a normalized even Majorana string up to phase. Since $\sum_i r_i=t(h)-m$, define the scalar integrand
\begin{equation}\label{eq:history-coefficient-v4}
c(h):=(-1)^m2^{\sum_i r_i}
\left(\prod_{i=1}^mz_{\lambda_*}(s_i)\right)
\left(\prod_{i=1}^m\prod_{j=1}^{r_i}z_{\nu_{i,j}}(\tau_{i,j})\right).
\end{equation}
Absorbing all Majorana multiplication phases into an even Majorana string up to a unit-modulus phase, denoted $Y_0(h)$, the history part of the expansion may be written compactly as
\begin{equation}\label{eq:history-integral-compact}
G(\tau_\beta)=I+\int_{\mathfrak H^\times}c(h)Y_0(h)\,d\pi(h),
\end{equation}
where $\mathfrak H^\times$ denotes the nonzero histories and $d\pi$ is counting measure on the discrete data times Lebesgue measure on the ordered time domains. By absorbing also the coefficient phase into
\begin{align}
    Y(h):=\frac{c(h)}{|c(h)|}Y_0(h),
\end{align}
we have
\begin{equation}\label{eq:history-integral-positive}
G(\tau_\beta)=I+\int_{\mathfrak H^\times}|c(h)|Y(h)\,d\pi(h).
\end{equation}
This is the form used below to construct the probability distribution.

\paragraph{5) one-block coefficient mass.}
Fix one root Dyson block with $r$ non-root labels. Its absolute coefficient mass is bounded by
\begin{align}\label{eq:block-mass-def-v4}
M_r:=\;2^r\sum_{\nu_1,\ldots,\nu_r\in\widehat S}
\int_0^{\tau_\beta} ds\,|z_{\lambda_*}(s)|
\int_{\Sigma_r(s)}\prod_{j=1}^r|z_{\nu_j}(\tau_j)|\,d\boldsymbol\tau \times\mathbf 1\!\left\{
\ad_{P_{\nu_r}}\cdots\ad_{P_{\nu_1}}(P_{\lambda_*})\ne0
\right\}.
\end{align}
For a nonzero nested commutator, the overlap graph on vertices $\{0,1,\ldots,r\}$, with $\nu_0:=\lambda_*$, is connected. Hence
\begin{equation}\label{eq:tree-domination-v4}
\mathbf 1_{\{\mathrm{nested}\ne0\}}
\le\sum_{\mathcal T\in\mathfrak T_{r+1}}
\prod_{\{p,q\}\in E(\mathcal T)}\mathbf 1_{\{\nu_p\sim\nu_q\}},
\end{equation}
where $\mathfrak T_{r+1}$ denotes the labeled spanning trees on $\{0,\ldots,r\}$.

We justify carefully the factor $1/r!$. After inserting \eqref{eq:tree-domination-v4} and summing over all labels $\nu_1,\ldots,\nu_r$ and all labeled trees, the resulting nonnegative integrand is invariant under a simultaneous permutation of the $r$ non-root pairs $(\tau_j,\nu_j)$, because the set of labeled trees is itself invariant under relabeling non-root vertices. Therefore
\[
\int_{\Sigma_r(s)}(\cdots)\,d\boldsymbol\tau
=\frac1{r!}\int_{[0,s]^r}(\cdots)\,d\boldsymbol\tau.
\]
Enlarging $[0,s]^r$ to $[0,\tau_\beta]^r$ then gives an upper bound. The root integration contributes
\[
\int_0^{\tau_\beta}|z_{\lambda_*}(s)|\,ds=\frac{a_{\lambda_*}}2.
\]
For a fixed rooted tree, orient every edge away from the root. Given a parent label $\mu$, summing over a child label and integrating its time variable, including the factor $2$ from the commutator, gives
\[
\sum_{\nu:\nu\sim\mu}2\int_0^{\tau_\beta}|z_\nu(t)|\,dt
=\sum_{\nu:\nu\sim\mu}a_\nu\le\Delta.
\]
Traversing the tree from the root therefore contributes at most $\Delta^r$. Cayley's formula gives $|\mathfrak T_{r+1}|=(r+1)^{r-1}$, hence
\begin{equation}\label{eq:block-mass-bound-v4}
M_r\le\frac{a_{\lambda_*}}2\frac{(r+1)^{r-1}}{r!}\Delta^r.
\end{equation}

\paragraph{6) normalization of one root Dyson block.}
Divide a root Dyson block of total length $r+1$ by $(3\Delta)^{r+1}$. Using $\theta=a_{\lambda_*}/\Delta$ and \eqref{eq:block-mass-bound-v4}, its normalized mass is at most
\[
\frac\theta2\frac{(r+1)^{r-1}}{r!3^{r+1}}.
\]
Let
\begin{align}
    \mathcal T(x):=\sum_{n\ge1}\frac{n^{n-1}}{n!}x^n
\end{align}
be the tree function. Summing over $r\ge0$ gives
\[
\frac\theta2\mathcal T(1/3)<\frac\theta2.
\]
To identify the power-series solution, note that for $n\ge4$ the ratio of consecutive summands is
\[
\frac{(n+1)^n}{3^{n+1}(n+1)!}\,
\frac{3^n n!}{n^{n-1}}
=\frac13\left(1+\frac1n\right)^{n-1}<\frac{10}{11}.
\]
Thus the tail from $n=5$ is bounded by the corresponding geometric series, and direct evaluation of the first four terms gives
\[
0<\mathcal T(1/3)\le\frac{4483}{5832}<1.
\]
This selects the solution analytic at zero; the functional equation $\mathcal T(x)e^{-\mathcal T(x)}=x$ alone would not distinguish it from the second positive solution.

\paragraph{7) several root Dyson blocks.}
The outer Dyson series imposes $0\le s_m\le\cdots\le s_1\le \tau_\beta$. For an upper bound on absolute mass, relax this ordering constraint and enlarge the outer simplex to $[0,\tau_\beta]^m$. The $m$ root Dyson blocks then factorize, so the total normalized mass of histories with $m$ root Dyson blocks is bounded by
\[
\left(\frac\theta2\mathcal T(1/3)\right)^m<\left(\frac\theta2\right)^m.
\]
The inner factor $1/r!$ came from permutation symmetry after tree domination; here we merely enlarge the outer simplex, so there is no additional $1/m!$.

\paragraph{8) probability measure on histories.}
Using the explicit history space from 4), define for every nonzero history
\begin{equation}\label{eq:history-probability-v4}
b(h):=\theta(3\Delta)^{t(h)},\qquad
dp(h):=\frac{|c(h)|}{\theta(3\Delta)^{t(h)}}\,d\pi(h).
\end{equation}
The block estimates imply
\[
\int_{\mathfrak H^\times}dp(h)
\le\frac1\theta\sum_{m\ge1}\left(\frac\theta2\right)^m
=\frac1{2(1-\theta/2)}\le1,
\]
where $0<\theta\le1$ because the definition of $\Delta$ includes the self-overlap activity $a_{\lambda_*}$. Assign the remaining probability mass to the zero-amplitude record $b=0$, $Y=I$, $\boldsymbol\lambda=\varnothing$. By \eqref{eq:history-integral-positive},
\begin{align*}
\mathbb E[bY]
=\int_{\mathfrak H^\times}\theta(3\Delta)^{t(h)}Y(h)
\frac{|c(h)|}{\theta(3\Delta)^{t(h)}}\,d\pi(h)
=G(\tau_\beta)-I.
\end{align*}
Since $G(\tau_\beta)=\mathcal R_{\lambda_*,S}$, this proves $\mathbb E[I+bY]=\mathcal R_{\lambda_*,S}$. Every nonzero record has rooted overlap history $\boldsymbol\lambda(h)$, and its history operator is the product of the listed Majorana strings up to a unit-modulus phase and satisfies
\[
b(h)=\theta(3\Delta)^{t(h)}.
\]
The distribution may equivalently be chosen countably supported. Let $\eta$ denote the finite label history, let $A_\eta$ be its exact integrated operator contribution, let $M_\eta$ be the integral of the absolute coefficient, and put $s_\eta:=\theta(3\Delta)^{t(\eta)}$. For $M_\eta>0$ set $p_\eta:=M_\eta/s_\eta$; if also $A_\eta\ne0$, set
\begin{align}
   b_\eta:=s_\eta\frac{\|A_\eta\|}{M_\eta},\qquad
Y_\eta:=\frac{A_\eta}{\|A_\eta\|}, 
\end{align}
while for $A_\eta=0$ take $b_\eta=0$ and $Y_\eta=I$. If $M_\eta=0$, take $p_\eta=b_\eta=0$ and $Y_\eta=I$. Each $A_\eta$ is a complex scalar times one fixed Majorana string, so every nonzero $Y_\eta$ is a history operator; moreover $p_\eta b_\eta Y_\eta=A_\eta$ and $b_\eta\le s_\eta$ because $\|A_\eta\|\le M_\eta$. The discrete histories form a countable union of finite Cartesian powers of $\mathcal L$, and absolute convergence permits the regrouping and gives $\sum_\eta p_\eta\le1$; the residual mass is again assigned to $b=0$, $Y=I$, $\boldsymbol\lambda=\varnothing$.
The theorem follows.
\end{proof}

\subsection{Support-decoupling construction}
\subsubsection{Overlap potential}
The support-decoupling construction needs more than a record of the
support of the current Majorana string. It must also measure how much of the
\emph{remaining} interaction-picture perturbation can still overlap that
string. This is the role of the overlap potential. For a Majorana string $X$ up to a unit-modulus phase, define
\begin{equation}
\Phi_S(X):=\frac1\Delta\sum_{\lambda\in S:\,\Supp(\lambda)\cap\Supp(X)\ne\varnothing}a_\lambda,
\end{equation}
with $\Phi_S(I)=0$.

\paragraph{Why this overlap potential is needed.}
There are two complementary uses of $\Phi_S$.  First, $\Phi_S(X)=0$ means that $X$ is disjoint from every remaining virtual label in $S$.  Since all relevant Majorana strings are even, $X$ then commutes with $V_S(t)$ and with every subsequent propagator generated by $V_S(t)$.  Thus a factor with zero overlap potential is \emph{decoupled}: later deletion steps leave it untouched.

Second, the overlap potential is a normalized potential controlling coefficient growth.  The support-decoupling invariant will be
\begin{equation}
|\alpha|\le 2^{-\Phi_S(X)}.
\end{equation}
If the next deleted label $\lambda_*$ overlaps $X$, then removing it lowers the overlap potential by
\begin{equation}
\theta:=\frac{a_{\lambda_*}}{\Delta},\qquad
\Phi_{S\setminus\{\lambda_*\}}(X)=\Phi_S(X)-\theta.
\end{equation}
The decomposition branch that leaves $X$ unchanged is chosen with probability $2^{-\theta}$, so conditioning on that branch enlarges its coefficient by $2^\theta$.  The decrease in the overlap potential exactly offsets this coefficient increase:
\begin{equation}
2^\theta 2^{-\Phi_S(X)}=2^{-\Phi_{S\setminus\{\lambda_*\}}(X)}.
\end{equation}
For decomposition branches that insert a propagator history, the history length $t$ will upper-bound the new overlap potential, while the smallness condition on the integrated overlap strength will make the reweighted coefficient at most $2^{-t}$.  Lemma~\ref{lem:potential-bounds} below is precisely the bridge between these two statements.

It is useful to introduce the overlap set
\begin{equation}
\mathcal C_S(X):=\{\lambda\in S:\ \Supp(\lambda)\cap\Supp(X)\ne\varnothing\}.
\end{equation}
Then
\begin{equation}
\Phi_S(X)=\frac1\Delta\sum_{\lambda\in\mathcal C_S(X)}a_\lambda.
\end{equation}

\begin{lemma}[Bounds on the overlap potential]\label{lem:potential-bounds}
For Majorana strings $X,Y$ up to unit-modulus phases,
\begin{equation}
\Phi_S(XY)\le\Phi_S(X)+\Phi_S(Y).
\end{equation}
More generally, if
\begin{equation}
Y=\zeta P_{\lambda_1}\cdots P_{\lambda_t},\qquad |\zeta|=1,
\end{equation}
for virtual labels $\lambda_1,\ldots,\lambda_t\in\mathcal L$, then
\begin{equation}
\Phi_S(Y)\le t.
\end{equation}
In particular, the latter estimate holds when $Y$ is generated by a rooted overlap history of length $t$.
\end{lemma}
\begin{proof}
We prove the two estimates separately.

\emph{Subadditivity.}  Multiplication of Majorana strings can cancel Majorana operators, but it cannot create a Majorana operator outside the union of the two input supports.  Hence
\begin{equation}
\Supp(XY)\subseteq\Supp(X)\cup\Supp(Y).
\end{equation}
If $\lambda\in\mathcal C_S(XY)$, choose
\[
x\in\Supp(\lambda)\cap\Supp(XY).
\]
The support inclusion implies that $x\in\Supp(X)$ or $x\in\Supp(Y)$.  Therefore $\lambda\in\mathcal C_S(X)$ or $\lambda\in\mathcal C_S(Y)$, and thus
\begin{equation}
\mathcal C_S(XY)\subseteq\mathcal C_S(X)\cup\mathcal C_S(Y).
\end{equation}
Since all activities $a_\lambda$ are nonnegative,
\begin{align}
\Delta\Phi_S(XY)
=\sum_{\lambda\in\mathcal C_S(XY)}a_\lambda
\le \sum_{\lambda\in\mathcal C_S(X)\cup\mathcal C_S(Y)}a_\lambda
\le \sum_{\lambda\in\mathcal C_S(X)}a_\lambda
   +\sum_{\lambda\in\mathcal C_S(Y)}a_\lambda
=\Delta\Phi_S(X)+\Delta\Phi_S(Y).
\end{align}
Dividing by $\Delta>0$ proves
\[
\Phi_S(XY)\le\Phi_S(X)+\Phi_S(Y).
\]
The same support argument also gives the useful identity
\begin{equation}
\Phi_S(X^\dagger)=\Phi_S(X).
\end{equation}

\emph{A product of $t$ virtual Majorana strings.}  Let
\[
Y=\zeta P_{\lambda_1}\cdots P_{\lambda_t}.
\]
Again, multiplication can only delete support, so
\begin{equation}
\Supp(Y)\subseteq\bigcup_{j=1}^t\Supp(\lambda_j).
\end{equation}
Suppose $\mu\in S$ overlaps $Y$.  Choose
\[
x\in\Supp(\mu)\cap\Supp(Y).
\]
The preceding inclusion implies that $x\in\Supp(\lambda_j)$ for some $j$, hence $\mu\sim\lambda_j$.  Therefore
\begin{equation}
\mathcal C_S(Y)\subseteq\bigcup_{j=1}^t
\{\mu\in S:\ \mu\sim\lambda_j\}.
\end{equation}
Using nonnegativity of the activities and then enlarging each inner sum from $S$ to the full virtual-label set $\mathcal L$ gives
\begin{align}
\Delta\Phi_S(Y)
&=\sum_{\mu\in\mathcal C_S(Y)}a_\mu
\le \sum_{j=1}^t\sum_{\substack{\mu\in S\\ \mu\sim\lambda_j}}a_\mu
\le \sum_{j=1}^t\sum_{\mu:\,\mu\sim\lambda_j}a_\mu.
\end{align}
By the definition of the integrated overlap strength,
\begin{equation}
\sum_{\mu:\,\mu\sim\lambda_j}a_\mu\le\Delta
\end{equation}
for every $j$.  Consequently
\begin{equation}
\Delta\Phi_S(Y)\le t\Delta,
\end{equation}
and therefore
\[
\Phi_S(Y)\le t.
\]
Notice that the rooted overlap condition is not needed for this overlap-potential estimate itself; it is needed earlier, in Theorem~\ref{thm:connected-propagator}, to obtain the coefficient bound for the sampled propagator history.
\end{proof}

\paragraph{Decoupled-factor consequence.}
If $\Phi_S(X)=0$ for an even Majorana string $X$ up to a unit-modulus phase, then $X$ has support disjoint from every label in $S$.  Indeed, every virtual label has positive activity, so a zero sum of activities contains no overlapping label.  Hence
\begin{equation}
[X,P_\lambda]=0\qquad(\lambda\in S),
\end{equation}
and therefore
\begin{equation}
[X,V_S(t)]=0,
\qquad
[X,U_S(t_2,t_1)]=0.
\end{equation}
This is the formal reason why the decoupled-factor condition in the
support-decoupling invariant permits all earlier factors to pass through every
subsequent deletion propagator unchanged.

\paragraph{Bridge between histories and support decoupling.}
A sampled history of length $t$ satisfies $\Phi_S(Y)\le t$. Together with
$b\le\theta(3\Delta)^t$, the decomposition-branch reweighting estimates below ensure that the resulting coefficient is at most
$2^{-t}\le2^{-\Phi_S(Y)}$. Separately, deleting $\lambda_*$ gives the exact relation
\[
\Phi_{S\setminus\{\lambda_*\}}(X)=\Phi_S(X)-\theta,
\qquad
2^\theta2^{-\Phi_S(X)}=2^{-\Phi_{S\setminus\{\lambda_*\}}(X)}.
\]
Without this overlap potential one could still write the history expansion, but there would be no inductive mechanism preventing repeated decomposition-branch reweighting from producing $|\alpha|>1$ before the support-decoupling construction terminates.

\subsubsection{Support-decoupling invariant}
At stage $k$, each record $\omega$ of the current distribution carries its own subset $S(\omega)\subseteq\mathcal L$ and finite ordered list
\[
((\alpha_1(\omega),X_1(\omega)),\ldots,(\alpha_m(\omega),X_m(\omega))).
\]
We maintain the following five invariants. Conditions (P2)--(P5) hold pointwise in $\omega$; in the displayed pointwise formulas we suppress $\omega$. In (P1), the expectation includes $\omega$ and all displayed data are evaluated at that record.

\begin{enumerate}[label=(P\arabic*)]
\item \textbf{Gibbs identity:}
\begin{equation}
e^{-\beta H}=\mathbb E\!\left[e^{-\beta H_0/2}U_S(\tau_\beta,0)\prod_{j=1}^m(I+\alpha_jX_j)U_S(0,-\tau_\beta)e^{-\beta H_0/2}\right].
\end{equation}

\item \textbf{Normalized Hermitian Majorana-string form:} $\alpha_j\in\mathbb R$ for every $j$. If $\alpha_j\ne0$, then
\[
X_j=\gamma_{K_j},\qquad X_j=X_j^\dagger,\qquad X_j^2=I,
\]
for an even index set $K_j$. An inactive factor is represented by $\alpha_j=0$ and $X_j=I$. Thus every real scalar, including signs and Hermitian-part factors $\operatorname{Re}(\zeta)$, is absorbed into $\alpha_j$.

\item \textbf{Disjoint active supports:} whenever $i\ne j$ and $\alpha_i\alpha_j\ne0$,
\[
\Supp(X_i)\cap\Supp(X_j)=\varnothing.
\]

\item \textbf{Decoupled-factor condition:} for every $j<m$ with $\alpha_j\ne0$,
\[
\Phi_S(X_j)=0.
\]

\item \textbf{Coefficient bound:}
\begin{equation}
|\alpha_j|\le2^{-\Phi_S(X_j)}.
\end{equation}
\end{enumerate}
The initial distribution is concentrated on one record with $S=\mathcal L$ and $m=0$, so (P1) follows from the symmetric interaction-picture identity.

\subsubsection{One deletion step}
For each current record $\omega$ with $S(\omega)\ne\varnothing$, if $m\ge1$, $\alpha_m(\omega)\ne0$, and $\Phi_{S(\omega)}(X_m(\omega))>0$, choose a selected label $\lambda_*(\omega)\in S(\omega)$ overlapping $X_m(\omega)$. Otherwise append the inactive factor $(0,I)$ and choose any selected label $\lambda_*(\omega)\in S(\omega)$. Suppressing $\omega$ in the following pointwise formulas, put
\[
\widehat S=S\setminus\{\lambda_*\},\qquad \theta=\frac{a_{\lambda_*}}{\Delta}.
\]
Independently sample $(b_1,Y_1)$ and $(b_2,Y_2)$ from the deletion-propagator theorem. Let $\alpha=\alpha_m$ and $X=X_m$. Choose decomposition branch 1 with probability
\[
p_1=2^{-\theta},
\]
and each decomposition branch $2,\ldots,7$ with probability
\[
p_j=\frac{1-2^{-\theta}}6.
\]
The successor distribution is the dependent union, over current records, of these two countably supported history samples and the seven-way decomposition-branch choice; every successor record inherits $\widehat S$ from the current record. Thus different current records may select different labels.
Define the Hermitian-part map
\[
\Herm(Y):=\frac{Y+Y^\dagger}{2}.
\]
The seven unnormalized Hermitian updates are
\begin{align*}
W_1&=\alpha X,\\
W_2&=b_1\Herm(Y_1),\qquad W_3=b_2\Herm(Y_2),\\
W_4&=b_1\alpha\Herm(Y_1X),\qquad W_5=b_2\alpha\Herm(Y_2X),\\
W_6&=b_1b_2\Herm(Y_1Y_2^\dagger),\\
W_7&=b_1b_2\alpha\Herm(Y_1XY_2^\dagger).
\end{align*}

Every argument of $\Herm$ above is an even Majorana string up to a unit-modulus phase. Thus, if $Y=\zeta \gamma_K$ with $|\zeta|=1$,
\begin{equation}
\Herm(Y)=\operatorname{Re}(\zeta)\gamma_K,\qquad
\operatorname{Re}(\zeta)\in[-1,1].
\end{equation}
For the selected decomposition branch $J$, write
\[
W_J=c_J\gamma_{K_J},\qquad c_J\in\mathbb R,
\]
with the convention $c_J=0$ and $\gamma_{K_J}=I$ if $W_J=0$. The updated factor is defined by
\begin{equation}
\widehat\alpha:=\frac{c_J}{p_J},\qquad \widehat X:=\gamma_{K_J}.
\end{equation}
In particular, the real-part factors are always absorbed into $\widehat\alpha$, never into $\widehat X$.

\begin{lemma}[Validity of one support-decoupling step]
Assume $\Delta\le1/72$. Applying the conditional successor construction above preserves probability normalization, the Gibbs-operator expectation identity, and the pointwise invariants (P2)--(P5).
\end{lemma}
\begin{proof}
The main issue is the coefficient invariant
\begin{equation}\label{eq:pinning-coeff-invariant-v4}
|\alpha_j|\le2^{-\Phi_S(X_j)}.
\end{equation}
A history of length $t$ can contribute at most $t$ units of new overlap
potential, so after decomposition-branch reweighting its coefficient must be no larger than
$2^{-t}$. The condition $\Delta\le1/72$ is a convenient sufficient condition
ensuring precisely this estimate.

\paragraph{1) algebraic form.}
For a history operator $Y=\zeta \gamma_K$, $|\zeta|=1$,
\[
\Herm(Y)=\frac{Y+Y^\dagger}{2}=\operatorname{Re}(\zeta)\gamma_K.
\]
Thus every nontrivial decomposition branch is a real scalar times a normalized Hermitian even Majorana string. By definition this real scalar is absorbed into $\widehat\alpha$, and $\widehat X$ remains normalized. This proves (P2), with no restriction $\zeta\in\{\pm1,\pm i\}$.

\paragraph{2) Gibbs identity.}
Writing $W_J=c_J\gamma_{K_J}$ does not change the selected operator product $\widehat\alpha\widehat X=W_J/p_J$. Hence averaging over the seven decomposition branches gives
\begin{align}\label{eq:seven-branch-average-v4}
\mathbb E_J[I+\widehat\alpha\widehat X]
=\frac12\Big[(I+b_1Y_1)(I+\alpha X)(I+b_2Y_2)^\dagger
+(I+b_2Y_2)(I+\alpha X)(I+b_1Y_1)^\dagger\Big].
\end{align}
Averaging over the two independent propagator samples gives, conditionally on the current record,
\[
\mathbb E[I+\widehat\alpha\widehat X]
=\mathcal R_{\lambda_*,S}(I+\alpha X)\mathcal R_{\lambda_*,S}^\dagger.
\]
Every previously decoupled factor has support disjoint from every label in $S$, hence commutes with $V_S(t)$, the corresponding propagators, $\mathcal R_{\lambda_*,S}$, and every sampled $Y_i$. Together with
\[
U_S(\tau_\beta,0)=U_{\widehat S}(\tau_\beta,0)\mathcal R_{\lambda_*,S},\qquad
U_S(0,-\tau_\beta)=\mathcal R_{\lambda_*,S}^\dagger U_{\widehat S}(0,-\tau_\beta),
\]
the tower property of expectation preserves (P1).

\paragraph{3) disjointness and decoupled factors.}
If $\Herm(Y)\ne0$, the support of the normalized Majorana string extracted from it is contained in $\Supp(Y)$. Each $Y_i$ is a product of virtual labels in $S$. A previously decoupled factor has support disjoint from every such label, so it is disjoint from each $Y_i$ and from every Majorana string appearing in the updated decomposition branch. Thus (P3) is preserved. Replacing $S$ by the smaller set $\widehat S$ preserves (P4) for all earlier decoupled factors.

\paragraph{4) origin of the constant $1/72$.}
The decomposition-branch probabilities are
\[
p_1=2^{-\theta},\qquad p_2=\cdots=p_7=\frac{1-2^{-\theta}}6,
\qquad 0<\theta\le1.
\]
Since $x\mapsto2^{-x}$ is convex on $[0,1]$,
\[
2^{-\theta}\le(1-\theta)2^0+\theta2^{-1}=1-\frac\theta2.
\]
Therefore
\begin{equation}\label{eq:branch-normalization-v4}
1-2^{-\theta}\ge\frac\theta2,
\qquad
\frac{6\theta}{1-2^{-\theta}}\le12.
\end{equation}
If one sampled propagator history has length $t\ge1$, Theorem~\ref{thm:connected-propagator} gives
\[
b\le\theta(3\Delta)^t.
\]
Thus on any nontrivial decomposition branch containing one propagator factor,
\begin{equation}\label{eq:single-history-reweight-v4}
\frac{b}{p_j}
\le\frac{6\theta}{1-2^{-\theta}}(3\Delta)^t
\le12(3\Delta)^t.
\end{equation}
The overlap-potential bound $\Phi_{\widehat S}(Y)\le t$ requires the right-hand side to be at most $2^{-t}$. It is sufficient that
\begin{equation}\label{eq:delta-condition-all-t-v4}
12(3\Delta)^t\le2^{-t}\qquad\text{for every }t\ge1.
\end{equation}
The most restrictive case is $t=1$, which gives
$36\Delta\le\frac12,$
i.e.
\begin{equation}\label{eq:delta-1-72-v4}
\Delta\le\frac1{72}.
\end{equation}
Conversely, under \eqref{eq:delta-1-72-v4}, $3\Delta\le1/24$, and hence for every $t\ge1$,
\begin{equation}\label{eq:two-to-minus-t-v4}
12(3\Delta)^t\le12\left(\frac1{24}\right)^t
=2^{-t}12^{1-t}\le2^{-t}.
\end{equation}
Thus the single-history reweighting bound is exactly small enough to preserve the coefficient bound in the support-decoupling invariant.

For decomposition branches containing both independent propagator samples, if their lengths are $t_1,t_2$ and $t_{\mathrm{tot}}=t_1+t_2$, then
\begin{align}\label{eq:double-history-reweight-v4}
\frac{b_1b_2}{p_j}
\le\frac{6\theta^2}{1-2^{-\theta}}(3\Delta)^{t_{\mathrm{tot}}}\le12\theta(3\Delta)^{t_{\mathrm{tot}}}
\le12(3\Delta)^{t_{\mathrm{tot}}}
\le2^{-t_{\mathrm{tot}}}.
\end{align}
Hence the same condition controls all six nontrivial decomposition branches.

\paragraph{5) decomposition branch 1.}
Suppose the last factor is active. By construction $\lambda_*$ overlaps $X$,
so deleting this one label removes exactly the activity
$a_{\lambda_*}/\Delta=\theta$ from its overlap potential:
\[
\Phi_{\widehat S}(X)=\Phi_S(X)-\theta.
\]
Decomposition branch 1 is selected with probability $2^{-\theta}$, so
\[
|\widehat\alpha|=2^\theta|\alpha|
\le2^{-\Phi_S(X)+\theta}
=2^{-\Phi_{\widehat S}(X)}.
\]
If the factor is inactive, its coefficient remains zero. Thus (P5) holds on decomposition branch 1. This explains the special choice $p_1=2^{-\theta}$: it exactly compensates for the decrease of the overlap potential by $\theta$.

\paragraph{6) history-only branches 2 and 3.}
Write
\[
\Herm(Y_i)=r_iX_i',\qquad |r_i|\le1.
\]
If the sampled history length is $t_i$, then by \eqref{eq:two-to-minus-t-v4},
\[
|\widehat\alpha|
=\frac{b_i|r_i|}{p_i}
\le2^{-t_i}.
\]
Because $\Supp(X_i')\subseteq\Supp(Y_i)$ and Lemma~\ref{lem:potential-bounds} gives $\Phi_{\widehat S}(Y_i)\le t_i$,
\[
\Phi_{\widehat S}(X_i')\le t_i.
\]
Therefore
\[
|\widehat\alpha|\le2^{-t_i}\le2^{-\Phi_{\widehat S}(X_i')},
\]
so (P5) holds.

\paragraph{7) history-and-factor branches 4 and 5.}
Write
\[
\Herm(Y_iX)=r_iX_i',\qquad |r_i|\le1.
\]
Then
\[
|\widehat\alpha|\le2^{-t_i}|\alpha|
\le2^{-t_i-\Phi_S(X)}.
\]
By subadditivity of the overlap potential,
\begin{align*}
\Phi_{\widehat S}(X_i')
\le\Phi_{\widehat S}(Y_iX)
\le\Phi_{\widehat S}(Y_i)+\Phi_{\widehat S}(X)
\le t_i+\Phi_S(X).
\end{align*}
Hence $|\widehat\alpha|\le2^{-\Phi_{\widehat S}(X_i')}$.

\paragraph{8) decomposition branch 6.}
Let $t_{\mathrm{tot}}=t_1+t_2$ and write
\[
\Herm(Y_1Y_2^\dagger)=rX',\qquad |r|\le1.
\]
Equation \eqref{eq:double-history-reweight-v4} gives
\[
|\widehat\alpha|\le2^{-t_{\mathrm{tot}}}.
\]
Moreover,
\[
\Phi_{\widehat S}(X')
\le\Phi_{\widehat S}(Y_1)+\Phi_{\widehat S}(Y_2)
\le t_1+t_2=t_{\mathrm{tot}}.
\]
Thus (P5) holds.

\paragraph{9) decomposition branch 7.}
Write
\[
\Herm(Y_1XY_2^\dagger)=rX',\qquad |r|\le1.
\]
Then
\[
|\widehat\alpha|
\le2^{-(t_1+t_2)}|\alpha|
\le2^{-t_1-t_2-\Phi_S(X)}.
\]
Also
\[
\Phi_{\widehat S}(X')
\le\Phi_{\widehat S}(Y_1)+\Phi_{\widehat S}(X)+\Phi_{\widehat S}(Y_2)
\le t_1+t_2+\Phi_S(X).
\]
Again (P5) follows. These real-part factors have absolute value at most one and cannot increase coefficient magnitudes, so all five invariants are preserved.
\end{proof}

\subsubsection{Support-decoupling construction with mapped Majorana supports}
The Hubbard application needs the preceding construction with overlaps measured
after several Majorana indices are grouped into one physical site.  We record
that variant here.

\begin{thm}[Support-decoupling construction with mapped Majorana supports]\label{prop:mapped-pinning}
Let $\pi:[2n]\to \mathcal X$ map Majorana indices to a finite target set.  For a
virtual label define
\[
 \Supp_\pi(\lambda):=\pi(\Supp(\lambda)),
\]
and put
\[
 \Delta_\pi:=
 \sup_\lambda\sum_{\substack{\nu:\,
 \Supp_\pi(\nu)\cap\Supp_\pi(\lambda)\ne\varnothing}}a_\nu.
\]
Assume time reflection and $\Delta_\pi\le1/72$.  Then
\[
 e^{-\beta H}
 =\sum_\omega p_\omega e^{-\beta H_0/2}
   \prod_j(I+\alpha_{\omega j}\gamma_{J_{\omega j}})
   e^{-\beta H_0/2},
\]
where $p_\omega\ge0$, the operator-valued sum is norm convergent,
$|\alpha_{\omega j}|\le1$, and the mapped supports
$\pi(J_{\omega j})$ are pairwise disjoint.
\end{thm}

\begin{proof}
Overlap of the original Majorana supports implies overlap of the mapped
supports, so the integrated overlap strength is at most $\Delta_\pi$. The deletion-propagator
expansion and its amplitude bound therefore remain valid with
$\Delta_\pi$.  For $\Delta_\pi>0$, replace the overlap potential by
\[
 \Phi^\pi_S(X)=\frac1{\Delta_\pi}
 \sum_{\lambda\in S:\,
   \Supp_\pi(\lambda)\cap\pi(\Supp(X))\ne\varnothing}a_\lambda.
\]
The inclusion
\[
 \pi(\Supp(XY))\subseteq
 \pi(\Supp(X))\cup\pi(\Supp(Y))
\]
gives the same subadditivity and history-length bounds as before; no
injectivity of $\pi$ is used.  If $\Phi^\pi_S(X)=0$, the mapped support of
$X$ is disjoint from every remaining label.  Its actual Majorana support is
therefore also disjoint from every remaining label, so $X$ commutes with the
remaining interaction and its propagators.  Thus (P1)--(P5), the exact
one-label deletion drop, and all seven decomposition-branch estimates hold with
$\Phi^\pi$ in place of $\Phi$.  The selected label is chosen separately for each
current record, and every step deletes that record's selected remaining
label.  Iteration therefore terminates after finitely many steps and gives
the displayed exact positive representation.

If the time radius is zero, the deterministic zero-time distribution is already
supported on a terminal record.  If the time radius is positive and $\Delta_\pi=0$, time reflection and membership in $\mathcal L$ imply that every virtual label has positive activity. Since each activity is at most $\Delta_\pi$, no virtual label exists, and the initial deterministic distribution is supported on a terminal record. These two cases never divide by
$\Delta_\pi$.
\end{proof}

\subsection{Completion of the sparse weak-coupling proof}
\begin{thm}[Integrated-overlap criterion]
Let $H=H_0+V$ be as above and let $\Delta$ be the integrated overlap strength
of the interaction-picture Majorana-basis expansion. If
\begin{equation}
\Delta\le\frac1{72},
\end{equation}
then $\rho_\beta(H)$ is a convex-Gaussian state.
\end{thm}
\begin{proof}
The case $\beta=0$ is immediate. For $\beta>0$, after $|\mathcal L|$ stages every terminal record $\omega$ satisfies $S(\omega)=\varnothing$, since each successor record deletes the selected remaining label of its current record. At termination, (P5) gives $|\alpha_j|\le1$ for every factor. By (P2), every active factor is exactly
\[
I+\alpha_j\gamma_{K_j},\qquad \alpha_j\in[-1,1],
\]
and by (P3) their supports are pairwise disjoint. Therefore the single-string criterion and Gaussian-cone closure give
\[
\prod_{j:\,\alpha_j\ne0}(I+\alpha_j\gamma_{K_j})\in\FGC.
\]
The final Gibbs invariant becomes
\[
e^{-\beta H}=\mathbb E\!\left[e^{-\beta H_0/2}\prod_{j:\,\alpha_j\ne0}(I+\alpha_j\gamma_{K_j})e^{-\beta H_0/2}\right].
\]
Gaussian congruence preserves $\FGC$. The terminal expectation is a norm-convergent positive sum or integral of cone elements, so the closedness established above gives $e^{-\beta H}\in\FGC$. Dividing by its positive trace gives the normalized claim.
\end{proof}

\begin{cor}[Local-strength criterion]
A sufficient condition for convex-Gaussianity is
\begin{equation}
2Rs_0\int_0^{\beta/2}e^{R\kappa_0t}\,dt\le\frac1{72}.
\end{equation}
For $\kappa_0>0$, equivalently,
\begin{equation}
\frac{2s_0}{\kappa_0}\left(e^{R\kappa_0\beta/2}-1\right)\le\frac1{72}.
\end{equation}
\end{cor}
\begin{proof}
Combine the bound on the integrated overlap strength with the integrated-overlap criterion.
\end{proof}

\subsubsection{Sparse parameter bound}
\begin{thm}[Sparse close-to-quadratic convex-Gaussian theorem]\label{thm:sparse-weak-sm}
Let $R,d\in\mathbb N$ and $\mathcal E,\epsilon\in\mathbb R$ satisfy $R\ge 2$, $\mathcal E>0$, and $\epsilon>0$. Let
\[
H=H_0+V=\sum_{b\in\mathcal A_0}q_b+\sum_{a\in\mathcal A_1}h_a
\]
be an even fermionic Hamiltonian on finitely many fermionic modes. Assume:
\begin{enumerate}[label=(\roman*)]
\item each $q_b$ is quadratic, Hermitian, and $\|q_b\|_{\gamma,1}\le \mathcal E$;
\item each $h_a$ is even, Hermitian, has Majorana support size at most $R$, and $\|h_a\|_{\gamma,1}\le \epsilon\mathcal E$;
\item the interaction graph on all local terms, with an edge whenever their Majorana supports intersect, has maximum degree at most $d$.
\end{enumerate}
If
\begin{equation}\label{eq:main-window-sm}
0\le\beta\le \frac{1}{R(d+1)\mathcal E}\log\!\left(1+\frac{1}{72\epsilon}\right),
\end{equation}
then
\[
\rho_\beta(H)=\frac{e^{-\beta H}}{\Tr{e^{-\beta H}}}
\]
is a convex-Gaussian state.
\end{thm}

\begin{proof}[Proof of Theorem~\ref{thm:sparse-weak-sm}]
By the sparse bounds,
\[
s_0\le(d+1)\epsilon\mathcal E,\qquad \kappa_0\le2(d+1)\mathcal E.
\]
Hence
\[
\Delta\le2R(d+1)\epsilon\mathcal E\int_0^{\beta/2}e^{2R(d+1)\mathcal Et}\,dt
=\epsilon\left(e^{R(d+1)\mathcal E\beta}-1\right).
\]
The assumed inverse-temperature window makes the right side at most $1/72$. Apply the integrated-overlap criterion.
\end{proof}

If $R=0$, $\mathcal E=0$, or $\epsilon=0$, the support and strength hypotheses force the perturbation to be scalar or zero. After scalar removal the Gibbs state is quadratic, and hence Gaussian, for every $\beta\ge0$.

\subsubsection{Weak-coupling lattice fermions}
Suppose $H_0$ contains local hopping and pairing terms with quadratic local scale $\mathcal E$, while $V$ contains local quartic terms with perturbation local scale $\epsilon\mathcal E$. If the interaction graph has maximum degree at most $d$ and every quartic term has Majorana support size at most $R=4$, then, for fixed interaction-graph degree bound $d$ and $\mathcal E$, the theorem certifies a convex-Gaussian inverse-temperature window whose upper endpoint scales as $\frac{\log(1/\epsilon)}{4(d+1)\mathcal E}$.
Unlike a global Dyson-norm criterion, this remains nontrivial for an extensive number of quartic terms.

\subsection{Logarithmic tightness of the weak-coupling inverse-temperature window}
\label{app:tightness}
\begin{proposition}[Asymptotic tightness]\label{prop:tight-sm}
There exists a family of $(R,d)$-sparse fermionic Hamiltonians with fixed Majorana support size bound $R=4$, interaction-graph degree bound $d=4$, quadratic local scale $\mathcal E=1$, and perturbation local scale $\epsilon$, such that for sufficiently small $\epsilon>0$ the Gibbs state is not convex-Gaussian whenever
\begin{align}
  \beta>b(\epsilon):=
\frac{1}{\sqrt{16+\epsilon^2}-\sqrt{4+\epsilon^2}}
\log\!\left(\frac{30}{1-4/\sqrt{16+\epsilon^2}}\right).  
\end{align}

Moreover,
\[
\lim_{\epsilon\to0^+}\left(b(\epsilon)-\log\frac1\epsilon\right)
=\frac12\log960.
\]
Consequently, for fixed Majorana support size and interaction-graph degree bound, the $\log(1/\epsilon)$ dependence in Theorem~\ref{thm:sparse-weak-sm} is asymptotically optimal up to universal constant factors.
\end{proposition}

\subsubsection{A sparse four-mode family}
\label{sec:matching-family}
We next give a family with fixed Majorana support size and interaction-graph degree bound showing that the logarithmic dependence on the perturbation Majorana coefficient norm bound is necessary. The only external ingredient is the complete four-mode convex-Gaussian criterion of Oszmaniec, Gutt, and Ku\'s~\cite{PhysRevA.90.020302}. For a pure state $|\psi\rangle$ in the positive-parity sector of four fermionic modes, their generalized concurrence is
\begin{equation}
C_+(|\psi\rangle)=|\langle\psi|\Theta_+|\psi\rangle|,
\end{equation}
and $C_+=0$ if and only if $|\psi\rangle$ is Gaussian. Moreover the maximal fidelity with the convex-Gaussian set is
\begin{equation}
F_{\rm Gauss}(|\psi\rangle)=\frac12\left(1+\sqrt{1-C_+(|\psi\rangle)^2}\right).
\label{eq:FG-fourmode}
\end{equation}
The antiunitary $\Theta_+$ implements excitation-hole duality; in particular its phase may be chosen so that
\begin{equation}
\Theta_+|0000\rangle=|1111\rangle,
\qquad
\Theta_+|1111\rangle=|0000\rangle.
\label{eq:theta-vac-full}
\end{equation}

\paragraph{The four-mode Hamiltonian.}
Take four complex fermionic modes with eight Majorana operators $\gamma_1,\ldots,\gamma_8$. We choose the convention
\[
\gamma_{2j-1}=c_j+c_j^\dagger,
\qquad
\gamma_{2j}=i(c_j^\dagger-c_j),
\]
so that
\[
Z_j:=i\gamma_{2j-1}\gamma_{2j}=2n_j-I,
\qquad j=1,\ldots,4.
\]
The four operators $Z_j$ are quadratic Majorana terms with pairwise disjoint supports. Define
\begin{equation}
H_0:=\sum_{j=1}^4 Z_j,
\qquad
Q:=\gamma_1\gamma_3\gamma_5\gamma_7,
\qquad
H_\epsilon:=H_0-\epsilon Q.
\label{eq:matching-gadget}
\end{equation}
The quartic operator is Hermitian, $Q^2=I$, and shares exactly one Majorana with every $Z_j$. Hence
\begin{equation}
\{Q,Z_j\}=0\qquad(j=1,\ldots,4),
\label{eq:Q-anticommutes}
\end{equation}
while the four $Z_j$ have mutually disjoint supports. The interaction graph is therefore a star graph with four leaves: $Q$ has degree $4$ and each quadratic dimer has degree $1$. Consequently the four-mode Hamiltonian is $(R,d)$-sparse with
\begin{equation}
R=4,\qquad d=4,\qquad \mathcal E=1,
\end{equation}
and the nonquadratic local scale is exactly $\epsilon$.

\paragraph{Ground state and non-Gaussianity.}
The states $|0000\rangle$ and $|1111\rangle$ have $H_0$-energies $-4$ and $+4$, respectively, and $Q$ flips all four occupations. Choosing the phase of $|1111\rangle$ so that $Q|0000\rangle=|1111\rangle$, the restriction of $H_\epsilon$ to their span is
\begin{equation}
\begin{pmatrix}
-4&-\epsilon\\
-\epsilon&4
\end{pmatrix}.
\label{eq:twolevel-gadget}
\end{equation}
Put
\[
E_0:=-\sqrt{16+\epsilon^2},
\qquad
E_1:=-\sqrt{4+\epsilon^2}.
\]
The unique ground state is
\begin{equation}
|g_\epsilon\rangle
=\cos\theta_\epsilon|0000\rangle
+\sin\theta_\epsilon|1111\rangle,
\label{eq:ground-gadget}
\end{equation}
where
\begin{equation}
\cos(2\theta_\epsilon)=\frac4{|E_0|},
\qquad
\sin(2\theta_\epsilon)=\frac{\epsilon}{|E_0|}.
\label{eq:mix-angle}
\end{equation}
Using~\eqref{eq:theta-vac-full},
\begin{equation}
C_+(g_\epsilon)
=2|\cos\theta_\epsilon\sin\theta_\epsilon|
=\frac{\epsilon}{\sqrt{16+\epsilon^2}}.
\label{eq:gadget-concurrence}
\end{equation}
Thus $|g_\epsilon\rangle$ is non-Gaussian for every $\epsilon>0$. Equation~\eqref{eq:FG-fourmode} gives the sharper bound
\begin{equation}
F_{\rm Gauss}(g_\epsilon)
=\frac12\left(1+\frac4{\sqrt{16+\epsilon^2}}\right),
\label{eq:gadget-fidelity}
\end{equation}
and hence
\begin{equation}
1-F_{\rm Gauss}(g_\epsilon)
=\frac12\left(1-\frac4{\sqrt{16+\epsilon^2}}\right)
=\frac{\epsilon^2}{64}+O(\epsilon^4).
\label{eq:gadget-deficit}
\end{equation}
For any convex-Gaussian state $\sigma$ on the full four-mode Fock space,
\begin{equation}
\langle g_\epsilon|\sigma|g_\epsilon\rangle
\le F_{\rm Gauss}(g_\epsilon).
\label{eq:gadget-witness}
\end{equation}
Indeed, odd-parity Gaussian components have zero overlap with $|g_\epsilon\rangle$, while each even-parity pure Gaussian component obeys~\eqref{eq:gadget-fidelity}; convexity then gives~\eqref{eq:gadget-witness}.

\paragraph{Spectral gap.}
Equation~\eqref{eq:Q-anticommutes} implies $\{H_0,Q\}=0$, and therefore
\[
H_\epsilon^2=H_0^2+\epsilon^2I.
\]
The possible absolute eigenvalues of $H_0$ are $4,2,0$. Hence the lowest two energies of $H_\epsilon$ are $E_0$ and $E_1$, and the spectral gap is
\begin{equation}
\Delta_\epsilon
:=E_1-E_0
=\sqrt{16+\epsilon^2}-\sqrt{4+\epsilon^2}
=2-\frac{\epsilon^2}{8}+O(\epsilon^4).
\label{eq:gadget-gap}
\end{equation}
In particular, the gap remains bounded away from zero as $\epsilon\to0^+$.

\paragraph{Thermal failure of convex-Gaussianity.}
Let
\[
\rho_{\beta,\epsilon}
:=\frac{e^{-\beta H_\epsilon}}{\Tr{e^{-\beta H_\epsilon}}},
\qquad
p_0(\beta,\epsilon)
:=\langle g_\epsilon|\rho_{\beta,\epsilon}|g_\epsilon\rangle.
\]
The four-mode Fock space has dimension $16$, so there are $15$ excited eigenstates, each at energy at least $\Delta_\epsilon$ above the ground state. Therefore
\begin{equation}
\frac{1-p_0}{p_0}
\le15e^{-\beta\Delta_\epsilon},
\qquad
p_0\ge1-15e^{-\beta\Delta_\epsilon}.
\label{eq:gadget-groundpopulation}
\end{equation}
If $\rho_{\beta,\epsilon}$ were convex-Gaussian,~\eqref{eq:gadget-witness} would imply $p_0\le F_{\rm Gauss}(g_\epsilon)$. Hence a sufficient condition for thermal failure of convex-Gaussianity is
\begin{equation}
15e^{-\beta\Delta_\epsilon}<1-F_{\rm Gauss}(g_\epsilon).
\label{eq:gadget-nongaussian-condition}
\end{equation}
Equivalently,
\begin{equation}
\beta>
\frac1{\Delta_\epsilon}
\log\!\left[
\frac{30}{1-4/\sqrt{16+\epsilon^2}}
\right]
\quad\Longrightarrow\quad
\rho_{\beta,\epsilon}\notin\operatorname{conv}(\FG)_4.
\label{eq:gadget-upper-threshold}
\end{equation}
Using~\eqref{eq:gadget-deficit} and~\eqref{eq:gadget-gap}, the right-hand side is
\begin{equation}
\log\frac1\epsilon+\frac12\log960+o(1)
\qquad(\epsilon\to0^+).
\label{eq:gadget-asymptotic}
\end{equation}
Thus a fixed $(R,d)=(4,4)$ sparse four-mode construction is not convex-Gaussian at inverse temperature $O(\log(1/\epsilon))$.

\paragraph{An arbitrarily large sparse family.}
For $N\ge1$, form the Hamiltonian $\sum_{\ell=1}^N H_{\epsilon,\ell}$ recursively using canonical ordered CAR embeddings. The parameters $R=4$, $d=4$, and $\mathcal E=1$ remain independent of $N$, while its Gibbs state is the ordered graded-CAR product of $N$ copies of $\rho_{\beta,\epsilon}$, not an ordinary tensor power. If this state were convex-Gaussian, its restriction to the last four-mode CAR subsystem would also be convex-Gaussian, but that restriction is the original one-block state. This contradicts~\eqref{eq:gadget-upper-threshold}. Therefore the same threshold for failure of convex-Gaussianity holds uniformly for all $N$.

Combining this family with Theorem~\ref{thm:sparse-weak-sm} shows that, for fixed Majorana support size and interaction-graph degree bound, the dependence of the universal convex-Gaussian inverse-temperature window on the perturbation Majorana coefficient norm bound is
\begin{equation}
\Theta(\log(1/\epsilon))
\end{equation}
up to nonoptimized constant factors.

\subsection{Weak-coupling Fermi--Hubbard specialization}
\label{app:weak-hubbard}
For a finite simple graph $G=(\Lambda,E)$ with $|\Lambda|=L$, consider the
centered spinful Fermi--Hubbard Hamiltonian
\[
H_{U,t}
=U\sum_{x\in\Lambda}
  \left(n_{x\uparrow}-\frac12\right)
  \left(n_{x\downarrow}-\frac12\right)
-t\sum_{\{x,y\}\in E}\sum_{\sigma\in\{\uparrow,\downarrow\}}
  \left(c_{x\sigma}^{\dagger}c_{y\sigma}
       +c_{y\sigma}^{\dagger}c_{x\sigma}\right).
\]
Define
\[
V_x:=(2n_{x\uparrow}-I)(2n_{x\downarrow}-I),
\qquad
H_{\rm hop}:=-\sum_{\{x,y\}\in E}\sum_{\sigma\in\{\uparrow,\downarrow\}}
  \left(c_{x\sigma}^{\dagger}c_{y\sigma}
       +c_{y\sigma}^{\dagger}c_{x\sigma}\right).
\]
With $V_{\rm int}:=\frac14\sum_{x\in\Lambda}V_x$, the centered Hamiltonian has the
exact decomposition
\[
H_{U,t}=tH_{\rm hop}+UV_{\rm int}.
\]

\subsubsection{Sparse-parameter verification}
Suppose that $\deg_G(x)\leq D$ for every $x\in\Lambda$.  For each undirected
edge $\{x,y\}\in E$ and spin $\sigma$, the corresponding local hopping term is
even, Hermitian, and quadratic.  Its Majorana expansion is supported on the
two Majoranas of the mode $(x,\sigma)$ and the two of $(y,\sigma)$, and its
Majorana coefficient norm is $|t|$.  The on-site interaction term $(U/4)V_x$ is even and
Hermitian, is supported on the four Majoranas at $x$, and has Majorana
coefficient norm at most $|U|/4$. Since the hopping sum is indexed by
undirected edges and
$V_x=4(n_{x\uparrow}-\tfrac12)(n_{x\downarrow}-\tfrac12)$, these local terms
sum exactly to the displayed $H_{U,t}$.

It remains to bound the interaction-graph degree. An on-site interaction term at $x$
overlaps at most two local hopping terms per edge incident on $x$, hence at most
$2D$ terms. A local hopping term on $\{x,y\}$ with spin $\sigma$ overlaps the two
on-site interaction terms at its endpoints and the other same-spin local hopping terms incident
on either endpoint. Their number is at most
$(\deg_G(x)-1)+(\deg_G(y)-1)+2\leq2D$. Thus $d\leq2D$.
The sparse theorem therefore applies with $R=4$, $\mathcal E=|t|$, $d=2D$,
and $\epsilon=|U|/(4|t|)$, proving convex-Gaussianity throughout the following inverse-temperature window:
\[
0\leq\beta\leq
\frac{1}{4(2D+1)|t|}
\log\!\left(1+\frac{|t|}{18|U|}\right).
\]
For fixed $D$ and $|U|/|t|\to0$, the upper endpoint of this sufficient
window scales as
\[
\Theta\!\left(\frac1{|t|}\log\frac{|t|}{|U|}\right).
\]
The theorem is a sufficient structural bound; its constants are not optimized.

\section{Proof of the strong-coupling Fermi--Hubbard theorem}
\label{app:strong-proof}
\begin{thm}[Strong-coupling Fermi--Hubbard theorem]
\label{thm:hubbard-strong-sm}
Let $D\in\mathbb N$, let $U,t,\beta\in\mathbb R$, and suppose
\[
\beta\geq0,\qquad U\neq0,\qquad \deg_G(x)\leq D
\quad\text{for every }x\in\Lambda.
\]
If
\begin{equation}
\frac{8D|t|}{|U|}
\left(e^{\beta|U|/2}-1\right)\leq\frac1{72},
\label{eq:hubbard-strong-condition-sm}
\end{equation}
then the Gibbs state $\rho_\beta(H_{U,t})$ is convex-Gaussian on
the $2L$ fermionic modes. In particular, if also $t\neq0$ and $D>0$, the
explicit inverse-temperature window
\begin{equation}
\beta\leq\frac{2}{|U|}
\log\!\left(1+\frac{|U|}{576D|t|}\right)
\label{eq:hubbard-strong-window-sm}
\end{equation}
implies the same conclusion.
\end{thm}

This criterion complements rather than replaces the preceding weak-coupling
estimate: it expands around the atomic-limit Hamiltonian and is useful in the
strong-coupling regime $|t|\ll|U|$.  Exact atomic imaginary-time conjugation bounds
the relevant site-overlap strength by the left side of
\eqref{eq:hubbard-strong-condition-sm}; the rooted overlap history estimate and
support-decoupling construction then reduce the Gibbs operator to terminal
dressed single-string factors with pairwise disjoint site supports.
The proof has four parts. We first reconstruct the Gibbs operator in the
atomic interaction picture and bound its site-overlap strength. We then prove
a site-deletion propagator theorem using rooted overlap histories. The atomic
Gaussian-cone lemmas identify the terminal factors that belong to $\FGC$.
Finally, a recordwise support-decoupling iteration gives an exact, norm-summable
terminal expansion. This last step requires an expectation invariant, not
only termination of label deletion. The atomic cone step is also necessary,
as congruence by a convex-Gaussian operator does not necessarily preserve
the Gaussian cone $\FGC$.

\subsection{Atomic interaction picture and site-overlap strength}
\label{sec:hubbard-strong-proof}
Recall that $G=(\Lambda,E)$ is a finite simple graph with $|\Lambda|=L$,
with two complex fermionic modes, or four Majorana indices, at each site.
With $n_{x\sigma}=c_{x\sigma}^{\dagger}c_{x\sigma}$, write
\[
 V_x=(2n_{x\uparrow}-I)(2n_{x\downarrow}-I),\qquad
 H_{\rm at}:=UV_{\rm int}=\frac U4\sum_xV_x,
 \qquad
 H_{U,t}=UV_{\rm int}+tH_{\rm hop}
 =H_{\rm at}+tH_{\rm hop}.
\]
Here $H_{\rm hop}=-\sum_{\{x,y\}\in E,\sigma}
(c_{x\sigma}^{\dagger}c_{y\sigma}+c_{y\sigma}^{\dagger}c_{x\sigma})$.
Each $V_x$ is the on-site parity operator, so $V_x^2=I$ and distinct $V_x$
commute. We use the normalized Hermitian Majorana basis fixed above;
in particular, $\gamma_K^{\dagger}=\gamma_K$ and $\gamma_K^2=I$.

\medskip
\noindent\textbf{Atomic interaction-picture expansion and virtual labels.}
The hats on $\widehat{h}$, $\widehat{U}$, $\widehat{\mathcal R}$, and
$\widehat{B}_*$ identify the strong-coupling interaction-picture objects,
distinguishing them from their weak-coupling counterparts.
The label set is denoted by $\mathcal L_{\mathrm{str}}$.
Put $\tau_\beta:=\beta/2$. For $x\ne y$ and
$\sigma\in\{\uparrow,\downarrow\}$, define
\[
 T_{xy,\sigma}:=c_{x\sigma}^{\dagger}c_{y\sigma},
 \qquad
 B_{xy,\sigma}:=T_{xy,\sigma}+T_{yx,\sigma}.
\]
The ordered bond-spin labels are
\[
 \mathcal A_{\rm hop}
 :=\bigl\{(x,y,\sigma)\in
 \Lambda\times\Lambda\times\{\uparrow,\downarrow\}:
 \{x,y\}\in E\bigr\}.
\]
For $a=(x,y,\sigma)\in\mathcal A_{\rm hop}$, define
\begin{equation}\label{eq:sm-atomic-local-term}
 \widehat{h}_a(s)
 :=-\frac t2 e^{sUV_{\rm int}}B_{xy,\sigma}e^{-sUV_{\rm int}}
 =\sum_{\substack{K\subseteq[4L]\\ |K|\text{ even}}}
 z_{a,K}(s)\gamma_K,
\end{equation}
where the last equality is the unique Majorana expansion. Define the active
virtual-label set by
\[
 \mathcal L_{\mathrm{str}}:=\bigl\{(a,K):a\in\mathcal A_{\rm hop},\ K\subseteq[4L],\
 |K|\text{ even},\ \exists s\in[-\tau_\beta,\tau_\beta]\text{ such that }
 z_{a,K}(s)\ne0\bigr\}.
\]
For $\lambda=(a,K)\in\mathcal L_{\mathrm{str}}$, put
\[
 P_\lambda:=\gamma_K,
 \qquad
 z_\lambda(s):=z_{a,K}(s),
 \qquad
 \Supp(\lambda):=K.
\]
The set $\mathcal A_{\rm hop}$ contains both orientations $(x,y)$ and
$(y,x)$ of every undirected bond, while
$B_{xy,\sigma}=B_{yx,\sigma}$. Thus the factor $1/2$ in
\eqref{eq:sm-atomic-local-term} exactly compensates for this double counting,
and, for $s\in[-\tau_\beta,\tau_\beta]$,
\begin{equation}\label{eq:sm-atomic-generator}
 V_{\mathcal L_{\mathrm{str}}}(s)
 :=\sum_{\lambda\in\mathcal L_{\mathrm{str}}}z_\lambda(s)P_\lambda
 =\sum_{a\in\mathcal A_{\rm hop}}\widehat{h}_a(s)
 =e^{sUV_{\rm int}}(tH_{\rm hop})e^{-sUV_{\rm int}}.
\end{equation}
For $s_1,s_2\in[-\tau_\beta,\tau_\beta]$, define its propagator by
\[
 \partial_{s_2}\widehat{U}_{\mathcal L_{\mathrm{str}}}(s_2,s_1)
 =-V_{\mathcal L_{\mathrm{str}}}(s_2)\widehat{U}_{\mathcal L_{\mathrm{str}}}(s_2,s_1),
 \qquad
 \widehat{U}_{\mathcal L_{\mathrm{str}}}(s_1,s_1)=I.
\]
Put $A_\beta:=e^{-\beta H_{\rm at}/2}$.
The following identity fixes the operator that the later iteration must
preserve.

\begin{proposition}[Atomic interaction-picture reconstruction]
\label{prop:sm-atomic-reconstruction}
For $s_1,s_2\in[-\tau_\beta,\tau_\beta]$,
\begin{equation}\label{eq:sm-atomic-propagator-explicit}
 \widehat{U}_{\mathcal L_{\mathrm{str}}}(s_2,s_1)
 =e^{s_2H_{\rm at}}e^{-(s_2-s_1)H_{U,t}}e^{-s_1H_{\rm at}}.
\end{equation}
Consequently,
\begin{equation}\label{eq:sm-atomic-symmetric-gibbs}
 e^{-\beta H_{U,t}}
 =e^{-\beta UV_{\rm int}/2}\widehat{U}_{\mathcal L_{\mathrm{str}}}(\tau_\beta,0)
  \widehat{U}_{\mathcal L_{\mathrm{str}}}(0,-\tau_\beta)e^{-\beta UV_{\rm int}/2}.
\end{equation}
The generator and its propagator also satisfy
\begin{equation}\label{eq:sm-atomic-time-reflection}
 V_{\mathcal L_{\mathrm{str}}}(s)^\dagger=V_{\mathcal L_{\mathrm{str}}}(-s),
 \qquad
 \widehat{U}_{\mathcal L_{\mathrm{str}}}(s_2,s_1)^\dagger
 =\widehat{U}_{\mathcal L_{\mathrm{str}}}(-s_1,-s_2).
\end{equation}
\end{proposition}
\begin{proof}
Differentiating the right side of
\eqref{eq:sm-atomic-propagator-explicit} cancels the $H_{\rm at}$ terms
and leaves $-V_{\mathcal L_{\mathrm{str}}}(s_2)$ times that expression, by
\eqref{eq:sm-atomic-generator}. At $s_2=s_1$ it is $I$, so uniqueness
of the finite-dimensional linear differential equation proves the formula.
Substituting $(s_2,s_1)=(\tau_\beta,0)$ and $(0,-\tau_\beta)$ gives
\[
 A_\beta \widehat{U}_{\mathcal L_{\mathrm{str}}}(\tau_\beta,0)=e^{-\tau_\beta H_{U,t}},
 \qquad
 \widehat{U}_{\mathcal L_{\mathrm{str}}}(0,-\tau_\beta)A_\beta=e^{-\tau_\beta H_{U,t}},
\]
which proves \eqref{eq:sm-atomic-symmetric-gibbs}.
Hermiticity of $H_{\rm at}$ and $B_{xy,\sigma}$ gives
$\widehat{h}_a(s)^\dagger=\widehat{h}_a(-s)$. Summing proves the generator reflection;
taking adjoints and reversing time in the propagator equation proves
the propagator reflection. In particular,
\[
 \widehat{U}_{\mathcal L_{\mathrm{str}}}(0,-\tau_\beta)
 =\widehat{U}_{\mathcal L_{\mathrm{str}}}(\tau_\beta,0)^\dagger.
\]
None of these identities requires the atomic Hamiltonian to be quadratic.
\end{proof}

\begin{lemma}[Atomic conjugation and local coefficient mass]
\label{lem:sm-atomic-local-mass}
For $a=(x,y,\sigma)$ and real $s$, the conjugated bond is supported on
$x,y$, its nonzero Majorana coefficients have degrees $2$, $4$, or $6$,
and
\[
 \|\widehat{h}_a(s)\|_{\gamma,1}\le\frac{|t|}{2}e^{|U||s|}.
\]
Its coefficient functions are continuous and satisfy
$z_{a,K}(-s)=\overline{z_{a,K}(s)}$.
\end{lemma}
\begin{proof}

For a directed hop as above, set
$N_{xy,\bar\sigma}=n_{x\bar\sigma}-n_{y\bar\sigma}$.
Here $\bar \sigma$ denotes the opposite of the spin $\sigma$.
The exact atomic
interaction-picture identity is
\begin{equation}\label{eq:sm-atomic-conjugation}
 e^{sUV_{\rm int}}T_{xy,\sigma}e^{-sUV_{\rm int}}
 =e^{sU N_{xy,\bar\sigma}}T_{xy,\sigma}.
\end{equation}
Indeed, the occupation basis diagonalizes both $UV_{\rm int}$ and
$N_{xy,\bar\sigma}$, so these two operators commute.  The canonical
anticommutation relations give
\begin{align}
     [UV_{\rm int},T_{xy,\sigma}]
   =U N_{xy,\bar\sigma}T_{xy,\sigma},
 \qquad
 [N_{xy,\bar\sigma},T_{xy,\sigma}]=0.
\end{align}
For commuting $A$ and $C$ satisfying $[A,B]=CB$ and $[C,B]=0$, the exponential
series gives $e^ABe^{-A}=e^CB$. Taking $A=sUV_{\rm int}$, $B=T_{xy,\sigma}$, and
$C=sU N_{xy,\bar\sigma}$ proves \eqref{eq:sm-atomic-conjugation}.

This identity supplies the only model-specific growth estimate needed by the
support-decoupling construction.  For completeness, let
\begin{align}
     \xi:=\frac{sU}{2},\qquad
 Z_{x\bar\sigma}:=2n_{x\bar\sigma}-I,\qquad
 Z_{y\bar\sigma}:=2n_{y\bar\sigma}-I,
\end{align}
and 
\begin{align}
     C_{xy,\sigma}:=i(T_{xy,\sigma}-T_{yx,\sigma}).
\end{align}
Combining the two directed conjugation formulas gives
\[
 e^{sUV_{\rm int}}B_{xy,\sigma}e^{-sUV_{\rm int}}
 =\cosh^2(\xi)B_{xy,\sigma}
  -\sinh^2(\xi)Z_{x\bar\sigma}Z_{y\bar\sigma}B_{xy,\sigma}
  -i\cosh(\xi)\sinh(\xi)
   (Z_{x\bar\sigma}-Z_{y\bar\sigma})C_{xy,\sigma}.
\]
In the normalized Majorana basis,
\[
 \|B_{xy,\sigma}\|_{\gamma,1}
 =\|C_{xy,\sigma}\|_{\gamma,1}
 =\|Z_{x\bar\sigma}\|_{\gamma,1}
 =\|Z_{y\bar\sigma}\|_{\gamma,1}=1.
\]
Submultiplicativity and the triangle inequality therefore give
\begin{align}
    \left\|e^{sUV_{\rm int}}B_{xy,\sigma}
 e^{-sUV_{\rm int}}\right\|_{\gamma,1}
 \le \cosh^2(\xi)+\sinh^2(\xi)
      +2|\cosh(\xi)\sinh(\xi)|
 =e^{2|\xi|}=e^{|U||s|}. 
\end{align}
It follows from the exact formula and the Majorana expansion in
\eqref{eq:sm-atomic-local-term} that, for every
$\lambda=((x,y,\sigma),K)\in\mathcal L_{\mathrm{str}}$, the operator $P_\lambda$ is even, is supported
on the two sites $x,y$, and has Majorana
degree $2$, $4$, or $6$. Since every $\gamma_K$ is Hermitian and its
Majorana expansion is unique, the termwise reflection
$\widehat{h}_a(s)^\dagger=\widehat{h}_a(-s)$ gives
\[
 z_\lambda(-s)=\overline{z_\lambda(s)},
 \qquad
 a_\lambda
 =\int_{-\tau_\beta}^{\tau_\beta}|z_\lambda(s)|\,ds
 =2\int_0^{\tau_\beta}|z_\lambda(s)|\,ds.
\]
Continuity follows either from the exponential formula or from the unique
finite Majorana expansion of the continuous matrix-valued function $\widehat{h}_a$.
If $\beta>0$, every active label has $a_\lambda>0$: a nonzero coefficient
value can be reflected to $[0,\tau_\beta]$, and continuity gives a
positive integral on a subinterval of positive length.
\end{proof}

Define the site map by
\[
 \pi:[4L]\longrightarrow\Lambda,
 \qquad
 \pi(j):=x\quad\text{when the Majorana index $j$ belongs to site $x$},
\]
and, as in Theorem~\ref{prop:mapped-pinning}, put
\[
 \Supp_\pi(\lambda):=\pi(\Supp(\lambda)).
\]
For a Majorana string $\gamma_J$, define its on-site Majorana degree at $x$ by
\[
 d_x(J):=\bigl|\{j\in J:\pi(j)=x\}\bigr|.
\]
Thus $|\Supp_\pi(\lambda)|\le2$. Define
\begin{equation}\label{eq:sm-site-degree-definition}
 \Delta_{\mathrm{site}}
 :=\max_{\lambda\in\mathcal L_{\mathrm{str}}}
 \sum_{\substack{\nu\in\mathcal L_{\mathrm{str}}:\\
 \Supp_\pi(\nu)\cap\Supp_\pi(\lambda)\ne\varnothing}}a_\nu,
\end{equation}
with value zero when $\mathcal L_{\mathrm{str}}$ is empty. Self-overlap is included.

\begin{proposition}[Site-overlap bound]
\label{prop:sm-site-overlap-bound}
If $\deg_G(x)\le D$ and $\beta\ge0$, then
\[
 \Delta_{\mathrm{site}}\le8D|t|\int_0^{\beta/2}e^{|U|s}\,ds.
\]
For $U\ne0$, this is at most
$\frac{8D|t|}{|U|}(e^{\beta|U|/2}-1)$.
\end{proposition}
\begin{proof}
Exactly $4\deg_G(x)$ ordered bond-spin
indices are incident on a site $x$: two orientations and two spins for every
neighboring vertex.  Each ordered term has prefactor $-t/2$. Grouping
Majorana coefficients by their ordered index can only enlarge their absolute
sum, and hence
\[
 \sum_{\substack{\lambda\in\mathcal L_{\mathrm{str}}:\,
                  \Supp(\lambda)\text{ contains a Majorana}\\
                  \text{index at site }x}}|z_\lambda(s)|
 \le4\deg_G(x)\frac{|t|}{2}e^{|U||s|}
 \le2D|t|e^{|U||s|}.
\]
For a fixed $\lambda$, summing over its at most two sites and then using the
incident coefficient mass bound gives
\begin{align*}
 \sum_{\substack{\nu\in\mathcal L_{\mathrm{str}}:\\
 \Supp_\pi(\nu)\cap\Supp_\pi(\lambda)\ne\varnothing}}a_\nu
 &\le
 \sum_{x\in\Supp_\pi(\lambda)}
 2\int_0^{\tau_\beta}
 \sum_{\nu:\,x\in\Supp_\pi(\nu)}|z_\nu(s)|\,ds\\
 &\le8D|t|\int_0^{\beta/2}e^{|U|s}\,ds.
\end{align*}
Taking the supremum over $\lambda$ gives, for $\beta\ge0$,
\begin{equation}\label{eq:sm-site-conflict-integral}
 \Delta_{\mathrm{site}}
 \le 8D|t|\int_0^{\beta/2}e^{|U|s}\,ds.
\end{equation}
If $U\ne0$, evaluation of the integral gives the closed form
\begin{equation}\label{eq:sm-site-conflict-closed}
 \Delta_{\mathrm{site}}
 \le \frac{8D|t|}{|U|}
     \left(e^{\beta|U|/2}-1\right).
\end{equation}
\end{proof}
At $\beta=0$ the site-overlap strength is exactly zero.  When $U=0$, formula
\eqref{eq:sm-site-conflict-closed} must not be used because of its denominator; the directly evaluated
version of \eqref{eq:sm-site-conflict-integral} is instead
$\Delta_{\mathrm{site}}\le4\beta D|t|$. Theorem~\ref{thm:hubbard-strong-sm}
assumes $U\ne0$; the zero-interaction model is covered separately
by the free-fermion case.

\subsection{Site-mapped deletion propagator and explicit random histories}
\label{sec:sm-site-propagator}
In this subsection assume $\beta>0$ and $\Delta_{\mathrm{site}}>0$,
where $\Delta_{\mathrm{site}}$ is the integrated site-conflict strength,
including self-conflicts.
The same proof allows any positive upper bound on that strength.
For every $S\subseteq\mathcal L_{\mathrm{str}}$, define
\[
 V_S(s):=\sum_{\lambda\in S}z_\lambda(s)P_\lambda,\qquad
 \partial_s \widehat{U}_S(s,u)=-V_S(s)\widehat{U}_S(s,u),\quad \widehat{U}_S(u,u)=I.
\]
Termwise time reflection gives
$V_S(s)^\dagger=V_S(-s)$ and
$\widehat{U}_S(s,u)^\dagger=\widehat{U}_S(-u,-s)$ for every $S$.
Fix $T:=\tau_\beta>0$, $\lambda_*\in S$, and
$\widehat S:=S\setminus\{\lambda_*\}$. Put
\[
 \widehat{\mathcal R}_{\lambda_*,S}:=\widehat{U}_{\widehat S}(0,T)\widehat{U}_S(T,0).
\]
This ordered deletion bridge is not assumed unitary. Composition and
reflection give
\begin{equation}\label{eq:sm-site-deletion-bridges}
 \widehat{U}_S(T,0)=\widehat{U}_{\widehat S}(T,0)\widehat{\mathcal R}_{\lambda_*,S},\qquad
 \widehat{U}_S(0,-T)=(\widehat{\mathcal R}_{\lambda_*,S})^{\dagger}\widehat{U}_{\widehat S}(0,-T).
\end{equation}

A block is a list $\boldsymbol\nu=(\nu_1,\ldots,\nu_r)$ in
$\widehat S$, preceded by one occurrence of $\lambda_*$. It is
site-rooted if each $\Supp_\pi(\nu_j)$ meets the root site support
or some earlier $\Supp_\pi(\nu_i)$, $i<j$.
A full discrete history is a nonempty ordered list of blocks,
\[
 h=(m;\boldsymbol\nu_1,\ldots,\boldsymbol\nu_m),\quad m\ge1,
 \qquad \ell(h):=\sum_i(1+r_i).
\]
Its flattened label list is
\[
 D(h)=(\lambda_*,\nu_{1,1},\ldots,\nu_{1,r_1},
            \lambda_*,\nu_{2,1},\ldots,\nu_{2,r_2},\ldots).
\]

\begin{thm}[Hubbard site-mapped deletion propagator]
\label{thm:sm-site-connected-propagator}
Assume $\beta>0$ and $\Delta_{\mathrm{site}}>0$.
For every $S\subseteq\mathcal L_{\mathrm{str}}$ and $\lambda_*\in S$,
$\theta:=a_{\lambda_*}/\Delta_{\mathrm{site}}$ belongs to $(0,1]$.
There is a countably supported probability distribution of
$(b,Y,h)$, including a residual atom $(0,I,\varnothing)$, such that
\[
 b\ge0,\qquad \mathbb E[I+bY]=\widehat{\mathcal R}_{\lambda_*,S},
\]
with absolute convergence in operator norm.
Every outcome with $b\ne0$ satisfies:
\begin{enumerate}[label=(\roman*)]
\item $\ell(h)\ge1$, all labels in $D(h)$ belong to $S$, and each
block is site-rooted. In fact, each block is rooted already under
intrinsic Majorana-support overlap.
\item For unit-modulus phases $\zeta_h,\xi_h$,
\[
 Y=\zeta_h\prod_{\lambda\text{ in }D(h)}P_\lambda
   =\xi_h\gamma_{K_h},\qquad
 K_h=\mathop{\triangle}_{\lambda\text{ in }D(h)}\Supp(\lambda),
\]
where $|K_h|$ is even and
\[
 \pi(K_h)\subseteq
 \bigcup_{\lambda\text{ in }D(h)}\Supp_\pi(\lambda).
\]
The phases need not be fourth roots of unity.
\item $b\le\theta(3\Delta_{\mathrm{site}})^{\ell(h)}$.
\end{enumerate}
The theorem itself requires no smallness bound $\Delta_{\mathrm{site}}\le1/72$.
\end{thm}

\begin{proof}
Continuity, reflection and activity somewhere in $[-T,T]$, with
$T>0$, imply
$a_\lambda=2\int_0^T|z_\lambda(s)|\,ds>0$.
Self-overlap gives
$0<a_{\lambda_*}\le\Delta_{\mathrm{site}}$.

Set
\[
 G_{\mathrm{str}}(u):=\widehat{U}_{\widehat S}(0,u)\widehat{U}_S(u,0),\qquad
 \widehat{B}_*(u):=\widehat{U}_{\widehat S}(0,u)P_{\lambda_*}\widehat{U}_{\widehat S}(u,0).
\]
Differentiation gives
$G_{\mathrm{str}}'(u)=-z_{\lambda_*}(u)\widehat{B}_*(u)G_{\mathrm{str}}(u)$ and $G_{\mathrm{str}}(0)=I$.
For a block define
\[
 C_{\boldsymbol\nu}:=
 \operatorname{ad}_{P_{\nu_r}}\cdots
 \operatorname{ad}_{P_{\nu_1}}(P_{\lambda_*}),\qquad
 \Sigma_r(s):=\{0\le\tau_r\le\cdots\le\tau_1\le s\}.
\]
At $r=0$ use $C_\varnothing=P_{\lambda_*}$ and inner integral one.
The first, conjugation expansion is
\[
 \widehat{B}_*(s)=\sum_{r\ge0}\sum_{\boldsymbol\nu\in\widehat S^r}
 \int_{\Sigma_r(s)}
 \left(\prod_{j=1}^r z_{\nu_j}(\tau_j)\right)
 C_{\boldsymbol\nu}\,d\boldsymbol\tau.
\]
Substituting it into the Volterra iteration of $G_{\mathrm{str}}$ gives the second,
outer expansion, and hence the exact two-stage Dyson series
\begin{align}\label{eq:sm-site-history-expansion}
 \widehat{\mathcal R}_{\lambda_*,S}&=I+\sum_h A_h,\\
 A_h&:=(-1)^m\int_{\mathfrak D_h(T)}
 \left(\prod_i z_{\lambda_*}(s_i)\right)
 \left(\prod_{i,j}z_{\nu_{i,j}}(\tau_{i,j})\right)
 C_{\boldsymbol\nu_1}\cdots C_{\boldsymbol\nu_m}
 \,d\boldsymbol\tau\,d\mathbf s,\nonumber\\
 \mathfrak D_h(T)&:=
 \{0\le s_m\le\cdots\le s_1\le T,\
       \boldsymbol\tau_i\in\Sigma_{r_i}(s_i),\ 1\le i\le m\}.
 \nonumber
\end{align}
Here the sum is over discrete label histories; times have already
been integrated. The finite-dimensional Dyson estimates give
$\sum_h\|A_h\|<\infty$: if
\[
 L_0:=2\sum_{\nu\in\widehat S}\int_0^T|z_\nu(u)|\,du,\qquad
 B_0:=e^{L_0}\int_0^T|z_{\lambda_*}(s)|\,ds,
\]
the inner absolute series is bounded by $e^{L_0}$, and the outer
Dyson majorant bounds $\sum_h\|A_h\|$ by $e^{B_0}-1$.
Thus the displayed regrouping by discrete histories is justified.

For a nonzero block, the CAR imply that its nested commutator is
$2^r$ times a unit-phased even Majorana string. A new intrinsic
support disjoint from the root and all preceding supports would
commute with their even product and annihilate the nested
commutator. This proves intrinsic rooted overlap, and hence site
rooted overlap. Products give the stated symmetric-difference
support. Applying $\pi$ only afterwards gives its inclusion in
the union of site supports.

Define the genuine absolute mass of a fixed block by
\[
 m(\boldsymbol\nu):=
 2^r\mathbf1_{\{C_{\boldsymbol\nu}\ne0\}}
 \int_0^T|z_{\lambda_*}(s)|
       \int_{\Sigma_r(s)}
             \prod_j|z_{\nu_j}(\tau_j)|
       \,d\boldsymbol\tau\,ds .
\]
The rooted-tree estimate \eqref{eq:block-mass-bound-v4} in the proof
of Theorem~\ref{thm:connected-propagator} gives
\begin{equation}\label{eq:sm-site-block-mass}
 \sum_{\boldsymbol\nu\in\widehat S^r}m(\boldsymbol\nu)
 \le\frac{a_{\lambda_*}}2
       \frac{(r+1)^{r-1}}{r!}\Delta_{\mathrm{site}}^r ,
\end{equation}
with combinatorial factor one at $r=0$.
Indeed, dominate the nonzero-commutator indicator by the sum over
spanning trees of its connected occurrence graph. After summing
labels and trees the non-root time-label pairs are symmetric,
so their ordered simplex contributes $1/r!$.
The root integral contributes $a_{\lambda_*}/2$.
For each tree edge, summing a child label and integrating its
time, with the factor two from the commutator, contributes at most
the intrinsic conflict degree, and thus at most $\Delta_{\mathrm{site}}$.
Cayley's tree count proves the estimate. These arguments require
no quadratic reference Hamiltonian. Site overlap is only a
majorant, not a sufficient test for a nonzero commutator.

Let
\[
 q:=\sum_{r\ge0}\sum_{\boldsymbol\nu\in\widehat S^r}
       \frac{m(\boldsymbol\nu)}{(3\Delta_{\mathrm{site}})^{r+1}}.
\]
With $\mathcal T(x)=\sum_{n\ge1}n^{n-1}x^n/n!$ as in the weak-coupling
proof, its previously established bound $\mathcal T(1/3)<1$ yields
$0\le q\le\frac{\theta}{2}\mathcal T(1/3)<\theta/2$.
For a full history put
\[
 M_h:=\prod_i m(\boldsymbol\nu_i),\qquad
 s_h:=\theta(3\Delta_{\mathrm{site}})^{\ell(h)}.
\]
Enlarging the ordered outer simplex to independent root-time
intervals gives $\|A_h\|\le M_h$, with no additional outer
$1/m!$. Define
\[
 p_h:=\frac{M_h}{s_h},\qquad
 b_h:=\begin{cases}s_h\|A_h\|/M_h,&M_h>0,\\0,&M_h=0,\end{cases}
 \qquad
 Y_h:=\begin{cases}A_h/\|A_h\|,&A_h\ne0,\\I,&A_h=0.\end{cases}
\]
Then $p_hb_hY_h=A_h$, $0\le b_h\le s_h$, and
\[
 \sum_h p_h=\frac1\theta\sum_{m\ge1}q^m
 \le\frac1\theta\sum_{m\ge1}(\theta/2)^m
 =\frac1{2(1-\theta/2)}\le1.
\]
Assign $p_0:=1-\sum_hp_h$ to $(0,I,\varnothing)$.
Every nonzero $A_h$ is a complex scalar times a single fixed
Majorana string, so its normalized $Y_h$ has the asserted form.
Finally,
\[
 \sum_h\|p_h(I+b_hY_h)\|
 \le\sum_hp_h+\sum_h\|A_h\|<\infty,\qquad
 p_0I+\sum_hp_h(I+b_hY_h)=I+\sum_hA_h=\widehat{\mathcal R}_{\lambda_*,S}.
\]
\end{proof}

\begin{lemma}[Site potential and sampled supports]
\label{lem:sm-site-history-potential}
For $S\subseteq\mathcal L_{\mathrm{str}}$ and a Majorana-index set $J$, define
\[
 \Phi^{\mathrm{site}}_S(J):=\frac1{\Delta_{\mathrm{site}}}
 \sum_{\substack{\lambda\in S:\,
 \Supp_\pi(\lambda)\cap\pi(J)\ne\varnothing}}a_\lambda.
\]
This potential is nonnegative, monotone in both arguments, and subadditive
on unions of index sets. Each virtual support has potential at most one.
For every sample in Theorem~\ref{thm:sm-site-connected-propagator},
using the identity monomial for zero amplitudes, one has
\begin{equation}\label{eq:sm-site-history-potential}
 \Phi^{\mathrm{site}}_{\widehat S}(\operatorname{supp}Y)\le\ell(h),
\end{equation}
where the residual history has length zero. If
$\Supp_\pi(\lambda_*)\cap\pi(J)\ne\varnothing$, then
\begin{equation}\label{eq:sm-site-potential-deletion}
 \Phi^{\mathrm{site}}_S(J)-\Phi^{\mathrm{site}}_{\widehat S}(J)=a_{\lambda_*}/\Delta_{\mathrm{site}}=\theta.
\end{equation}
\end{lemma}
\begin{proof}
All activities are nonnegative. A label meeting the site image of a union
meets at least one of the individual site images, which proves subadditivity.
The defining conflict-degree bound gives
$\Phi^{\mathrm{site}}_S(\Supp(\nu))\le\Delta_{\mathrm{site}}/\Delta_{\mathrm{site}}=1$ for each label
$\nu$. The sampled support is contained, after applying $\pi$, in the
union of the $\ell(h)$ virtual site supports. Summing their potentials
therefore proves~\eqref{eq:sm-site-history-potential}. If the amplitude
vanishes, choosing $Y=I$ gives potential zero directly. Finally, deleting
one overlapping label removes precisely its activity from the potential
sum, proving~\eqref{eq:sm-site-potential-deletion}.
\end{proof}

\subsection{Atomic Gaussian-cone lemmas}
A key property used in the weak-coupling regime is the closure of the Gaussian cone under Gaussian congruences: if $A \in \text{FGC}$, then $e^{-Q/2} A e^{-Q/2} \in \text{FGC}$ for any quadratic Hamiltonian $Q$. However, congruence by a convex-Gaussian operator does not necessarily preserve the Gaussian cone $\FGC$. Consequently, analyzing the strong-coupling regime requires additional lemmas that exploit the specific structure of the Fermi–Hubbard model.
We prove the required atomic statement locally and then combine factors
with disjoint site supports. In particular, no general cone-preservation
claim for the quartic atomic congruence is used.

All Majorana strings below use the fixed Hermitian normalization introduced
above: $\gamma_J^\dagger=\gamma_J$ and $\gamma_J^2=I$.
We keep track of the full site carrier of every local cone decomposition,
namely all Majorana labels on its one or two sites.  The local operators are
embedded into the full system by the ordered CAR embeddings.  The supported
form of disjoint-support closure then allows us to multiply cone elements
on distinct site blocks.  In particular, the factors being interchanged
are even and have disjoint site carriers; this argument does not assume
that arbitrary products of cone elements belong to $\FGC$.

\begin{lemma}[\bf Odd-site Gaussian-cone lemma]
\label{lem:odd-site}
Consider four modes belonging to two Hubbard sites, and let $\gamma_J$ be the
normalized Hermitian Majorana string associated with $J$; in particular,
$\gamma_J^\dagger=\gamma_J$ and $\gamma_J^2=I$. Suppose
the two sites are $x,y$ and that $d_x(J)$ and $d_y(J)$ are odd. For every
$a_x,a_y,r\in\mathbb R$ with $|r|\le1$,
\begin{equation}\label{eq:sm-two-site-cone}
 e^{-a_xV_x-a_yV_y}+r \gamma_J\in\FGC.
\end{equation}
\end{lemma}
\begin{proof}
    
For an occupation $f\in\{0,1\}^4$, define the occupation $g(f)$ and phase
$\zeta(f)$ by
\[
 \gamma_J|f\rangle=\zeta(f)|g(f)\rangle.
\]
Unitarity gives $|\zeta(f)|=1$, and applying $\gamma_J^2=I$
shows that $g(g(f))=f$ and $\zeta(g(f))\zeta(f)=1$.
In particular, $\zeta(g(f))=\overline{\zeta(f)}$.
The occupation map has no fixed points, as the following flip count shows,
so these data decompose the occupation basis into two-element orbits. The odd
values of $d_x(J)$ and $d_y(J)$ imply
\[
 g(g(f))=f,
 \qquad |\{k:f_k\ne g(f)_k\}|=2,
 \qquad |\zeta(f)|=1,
 \qquad \zeta(g(f))=\overline{\zeta(f)}.
\]
Indeed, an odd Majorana string on the two modes of one site has degree $1$ or
$3$.  A degree-$3$ string is, up to a phase, the site-parity operator times
the other degree-$1$ Majorana operator. In either case it flips exactly
one occupation on that site. Therefore $\gamma_J$ flips one occupation at
each of the two sites, which proves the second identity above.
If $p_x(f),p_y(f)\in\{\pm1\}$ are the two site-parity eigenvalues, put
\[
 w(f)=e^{-a_xp_x(f)-a_yp_y(f)}.
\]
The two parities change sign under $g$, and hence
\begin{equation}\label{eq:sm-weight-reflection}
 w(g(f))=w(f)^{-1},\qquad w(f)>0.
\end{equation}

We use the following four-mode two-determinant criterion. If occupation configurations
$f$ and $g$ differ in exactly two modes, then for all $a,b\in\mathbb C$,
\[
 |a f+b g\rangle\langle a f+b g|\in\FGC.
\]
Indeed, fix the occupations of the two modes on which $f$ and $g$ agree and relabel the two
differing modes.  If the active occupations are $00$ and $11$, the normalized
nonzero vector is a two-mode BCS state
\[
 a|00\rangle+b|11\rangle
 \propto\exp(\lambda c_1^\dagger c_2^\dagger)|00\rangle,
\]
with zero-coefficient endpoints obtained directly or by continuity.  If the
active occupations are $10$ and $01$, it is the one-particle Slater
determinant
\[
 (a c_1^\dagger+b c_2^\dagger)|00\rangle.
\]
Both are pure fermionic Gaussian states.  The fixed-mode occupation factor is
Gaussian, and ordered CAR products and mode relabelings preserve Gaussianity.
Multiplying the normalized projector by the squared norm proves the
unnormalized cone statement; the zero vector is immediate.

Define the two-determinant block on the two-element occupation orbit $\{f,g(f)\}$ by

 \begin{align}
 B_f=w(f)|f\rangle\langle f|
       +w(f)^{-1}|g(f)\rangle\langle g(f)|
     +r \zeta(f)|g(f)\rangle\langle f|
       +r\overline{\zeta(f)}|f\rangle\langle g(f)|.
 \end{align}
Each $B_f$ lies in $\FGC$.  To see the underlying factorization explicitly,
write $s_f=\sqrt{w(f)}$.  Then
\[
 B_f=|v_f\rangle\langle v_f|+|u_f\rangle\langle u_f|,
 \quad
 v_f=s_f|f\rangle+\frac{r \zeta(f)}{s_f}|g(f)\rangle,
 \quad
 u_f=\frac{\sqrt{1-r^2}}{s_f}|g(f)\rangle.
\]
The condition $|r|\le1$ ensures $1-r^2\ge0$, while $s_f>0$
makes both denominators well defined.  At $|r|=1$ the vector $u_f$ vanishes,
which is permitted. The two occupations differ in exactly two entries, so
both rank-one terms are Gaussian by the four-mode two-determinant criterion;
their nonnegative sum is therefore in the cone.

The occupation basis diagonalizes $-a_xV_x-a_yV_y$, and the Majorana action has
phase $\zeta(f)$.
Consequently
\[
 e^{-a_xV_x-a_yV_y}=\sum_f w(f)|f\rangle\langle f|,
 \qquad
 r \gamma_J=\sum_f r \zeta(f)|g(f)\rangle\langle f|.
\]
Reindexing both reverse terms along the involution $g$ and using \eqref{eq:sm-weight-reflection} gives
\[
 e^{-a_xV_x-a_yV_y}+r \gamma_J=\frac12\sum_f B_f.
\]
Thus \eqref{eq:sm-two-site-cone} follows from closure of $\FGC$ under finite sums and nonnegative
scaling.  This is the local step that treats paired sites of odd on-site Majorana degree.

\end{proof}

We also need the two-mode cone criterion for positive even operators.

\begin{lemma}\label{lem:onesite}
Let $A$ be a positive operator on the two complex fermionic modes of one Hubbard site.  If $[A,V_x]=0$, then $A\in\FGC$.
\end{lemma}
\begin{proof}
The local Fock space splits into two parity sectors,\[
\mathcal H_{\mathrm{even}}=\operatorname{span}\{|00\rangle,|11\rangle\},
\qquad
\mathcal H_{\mathrm{odd}}=\operatorname{span}\{|10\rangle,|01\rangle\}.
\]
Since $A$ commutes with $V_x$, it is block diagonal in these sectors.  Spectrally decompose each positive $2\times2$ block.  Every pure even state $a|00\rangle+b|11\rangle$ is a two-mode BCS Gaussian state, and every pure odd state $a|10\rangle+b|01\rangle$ is a one-particle Slater determinant.  Thus every eigenprojector is Gaussian, and the positive spectral decomposition places $A$ in $\FGC$.
\end{proof}

\begin{lemma}[\bf Even-site Gaussian-cone lemma]
\label{lem:even-site}
If $J_x$ is supported on one site and has even on-site Majorana
degree, then for every $a_x,r\in\mathbb R$ with $|r|\le1$,
\begin{align}
     e^{-a_xV_x/2}(I+r \gamma_{J_x})e^{-a_xV_x/2}\in\FGC.
\end{align}

\end{lemma}

\begin{proof}
   Indeed, the Hermitian involution $\gamma_{J_x}$ has spectrum in
$\{-1,1\}$, so $I+r\gamma_{J_x}\ge0$ for $|r|\le1$.
The factor $D_x=e^{-a_xV_x/2}$ is Hermitian, and hence
$D_x(I+r\gamma_{J_x})D_x$ is positive.
Moreover, as $J_x$ has even on-site Majorana
degree, we have $[\gamma_{J_x},V_x]=0$ and thus 
$[I+r \gamma_{J_x}, V_x]=0$.
Since $D_x$ is a function of $V_x$, it too commutes with $V_x$;
therefore the entire positive sandwich is even on this site.
Lemma~\ref{lem:onesite} now applies.  Its ordered CAR embedding gives a
cone decomposition supported on the full carrier of the site.
In particular, taking $r=0$ proves $e^{-a_xV_x}=D_x^2\in\FGC$,
which will supply the factors on unused sites.
\end{proof}

We next assemble the general dressed single-string factor. 
\begin{lemma}[\bf Atomic Hubbard Gaussian-cone lemma]
\label{lem:sm-atomic-dressed-string}
    For real numbers $a_x$, a set
$S$ of sites, and a normalized Hermitian even Majorana string $\gamma_J$
supported on those sites, define the touched-site set
\[
 S_J:=\{x\in S:\text{$J$ contains at least one of the four Majorana indices at $x$}\}.
\]
For every $A\subseteq S$, let
\begin{align}
    \mathcal D_A:=\prod_{x\in A}e^{-a_xV_x/2}.
\end{align}
Then, for every real $\alpha$ with $|\alpha|\le1$,
\begin{equation}\label{eq:sm-atomic-sandwich}
 \mathcal D_S(I+\alpha \gamma_J)\mathcal D_S\in\FGC.
\end{equation}
\end{lemma}
\begin{proof}
    For nonempty $J$, the sum of the on-site Majorana degree over all sites is $|J|$, which is
even. Hence the number of sites with odd on-site degree is even.
For a deterministic partition, first list the even-degree touched sites in
increasing order as singleton blocks.  List the odd-degree sites in increasing
order and pair consecutive sites in that list, placing these pairs after
the singleton blocks.  No odd site remains unpaired.  Thus the blocks cover
$S_J$ exactly once, their full site carriers are pairwise disjoint, and
the restricted Majorana supports of all blocks have even cardinality.
For each block $i$, let
\begin{align}
     J_i:=\{j\in J:\pi(j)\in i\},
\end{align}
and let $\gamma_{J_i}$ use the same fixed Hermitian normalization as
$\gamma_J$.  We do not choose its phase independently; the ordered product
of these local strings can differ from the global string by a sign,
which will be retained explicitly below.

For an even singleton block $i=\{x\}$, let
\begin{align}
    B_i:=e^{-a_xV_x},
 \qquad
 C_i:=e^{-a_xV_x/2}\gamma_{J_i}e^{-a_xV_x/2}.
\end{align}
The local string commutes with $V_x$, and thus Lemma~\ref{lem:even-site}
gives
\[
 B_i+rC_i
 =e^{-a_xV_x/2}(I+r\gamma_{J_i})e^{-a_xV_x/2}\in\FGC,
 \qquad |r|\le1.
\]
For an odd two-site block $i=\{x,y\}$, let
\begin{align}
     B_i:=e^{-a_xV_x-a_yV_y},
 \qquad C_i:=\gamma_{J_i}.
\end{align}
The local string anticommutes with both $V_x$ and $V_y$.
Writing $Q=(a_xV_x+a_yV_y)/2$, this gives
$e^{-Q}\gamma_{J_i}=\gamma_{J_i}e^{Q}$, and consequently
\[
 e^{-(a_xV_x+a_yV_y)/2}\gamma_{J_i}
 e^{-(a_xV_x+a_yV_y)/2}=C_i.
\]
Then Lemma~\ref{lem:odd-site} gives
\[
 B_i+rC_i\in\FGC,
 \qquad |r|\le1.
\]

Since $J\ne\varnothing$, there are $m\ge1$ blocks.  Reordering the
product of their normalized Majorana strings into the ordered word for $J$
gives a scalar phase multiplying $\gamma_J$.  Every block string is even;
therefore block strings on disjoint carriers commute, and their product is
both Hermitian and an involution.  The scalar phase is consequently real
and equals a sign $\eta\in\{\pm1\}$.
The atomic half-density factors are also even and commute across distinct
blocks.  Moving only these cross-block factors, without commuting a local
string through its own half-density, gives
\begin{align}
     \prod_{i=1}^m B_i=\mathcal D_{S_J}^2,
 \qquad
 \prod_{i=1}^m C_i=\eta\,\mathcal D_{S_J}\gamma_J\mathcal D_{S_J}.
\end{align}
Let $\varrho=|\alpha|^{1/m}\le1$.  If $\alpha\ne0$, put
$\tau=\eta\operatorname{sgn}(\alpha)$; if $\alpha=0$, put $\tau=1$.
Choose $\chi_1,\ldots,\chi_{m-1}$ independently and uniformly from
$\{\pm1\}$ and let
\begin{align}
     \chi_m=\tau\prod_{i=1}^{m-1}\chi_i.
\end{align}
For $m=1$ this prescription means simply $\chi_1=\tau$, with no
independent signs to sample.  For a nonempty proper subset $A$, if $m\notin A$,
one of the independent signs in the subproduct has zero mean.  If $m\in A$,
substituting for $\chi_m$ leaves an independent sign indexed outside $A$
to an odd power, and averaging that sign again gives zero.  Thus
\begin{align}
   \mathbb E\!\left[\prod_{i\in A}\chi_i\right]=0
 \quad\bigl(\varnothing\ne A\subsetneq\{1,\ldots,m\}\bigr),
 \qquad
 \mathbb E\!\left[\prod_{i=1}^m\chi_i\right]=\tau. 
\end{align}
Expanding in the fixed order therefore gives
\begin{align}
 \mathbb E_\chi\prod_{i=1}^m(B_i+\varrho\chi_iC_i)
 =\prod_{i=1}^mB_i+\tau \varrho^m\prod_{i=1}^mC_i
 =\mathcal D_{S_J}(I+\alpha \gamma_J)\mathcal D_{S_J}.
\end{align}
Here $\varrho^m=|\alpha|$, and if $\alpha=0$ then $\varrho=0$,
so the identity holds for the chosen $\tau=1$ as well.
The expansion uses a fixed operator order and requires no commutation
between $B_i$ and $C_i$ within a block.
Every product inside the expectation belongs to $\FGC$ by
Lemmas~\ref{lem:odd-site} and~\ref{lem:even-site} and supported
disjoint-CAR-product closure.  The expectation is a finite average with
nonnegative weights $2^{-(m-1)}$. Taking this positive average proves
the required dressed single-string factor on $S_J$ belongs to $\FGC$.

There are two edge cases in \eqref{eq:sm-atomic-sandwich}, both needed later.  If $J=\varnothing$,
then the dressed single-string factor equals $(1+\alpha)\mathcal D_S^2$.
Each onsite density $e^{-a_xV_x}$ is in the one-site cone by the
$r=0$ case of Lemma~\ref{lem:even-site}; their product
$\mathcal D_S^2$ is in the cone by disjoint-support closure.
Since $1+\alpha\ge0$, the empty-string factor is also in the cone,
including $\alpha=-1$, when it vanishes.  If $S=\varnothing$,
the onsite product is the identity, so the same argument applies.
If $S_J\subsetneq S$, operators on the unused sites commute
with both the atomic half-Gibbs factor on $S_J$ and the Majorana string, and one has the exact
factorization
\begin{align}
    \mathcal D_S(I+\alpha \gamma_J)\mathcal D_S
 =\mathcal D_{S_J}(I+\alpha \gamma_J)\mathcal D_{S_J}\,\mathcal D_{S\setminus S_J}^{\,2}.
\end{align}
The second factor is the square of an atomic half-Gibbs factor and belongs to the
cone on its disjoint support. 
Hence, we prove  \eqref{eq:sm-atomic-sandwich} for arbitrary $S$.

\end{proof}

\begin{cor}
\label{cor:sm-terminal-atomic-cone}
    Let $(J_j,\alpha_j)_{j=1}^m$ be a finite family such that every $J_j$ has
even cardinality, $\alpha_j\in\mathbb R$, $|\alpha_j|\le1$, and the nonempty
$J_j$ have pairwise disjoint site supports.
Then 
\begin{equation}\label{eq:sm-terminal-sandwich}
 e^{-\beta UV_{\rm int}/2}
 \prod_j(I+\alpha_j \gamma_{J_j})
 e^{-\beta UV_{\rm int}/2}\in\FGC.
\end{equation}

\end{cor}

\begin{proof}
Let $E:=\{j:J_j=\varnothing\}$ and $N:=\{j:J_j\ne\varnothing\}$, and set
\[
 c_E:=\prod_{j\in E}(1+\alpha_j)\ge0.
\]
Each empty string contributes its scalar factor to $c_E$.
For each $j\in N$, put
\[
 S_j:=\pi(J_j)=\{\pi(q):q\in J_j\},
 \qquad S_0:=\Lambda\setminus\bigcup_{j\in N}S_j.
\]
The sets $S_j$ for $j\in N$ are pairwise disjoint, and $S_0$ contains
the unused sites.  Taking $a_x=\beta U/4$ in the definition of
$\mathcal D_A$ and using commutation of the onsite parity operators gives
\[
 \mathcal D_\Lambda
 =\prod_{x\in\Lambda}e^{-\beta UV_x/8}
 =e^{-\beta UV_{\rm int}/2}.
\]
The local centers and the atomic half-densities are even and supported
on their indicated site carriers. Thus factors from distinct site sets
commute, yielding the exact factorization
\[
 e^{-\beta UV_{\rm int}/2}
 \prod_{j=1}^{m}(I+\alpha_j\gamma_{J_j})
 e^{-\beta UV_{\rm int}/2}
 =c_E\,\mathcal D_{S_0}^{2}
 \prod_{j\in N}\left[
  \mathcal D_{S_j}(I+\alpha_j\gamma_{J_j})\mathcal D_{S_j}
 \right],
\]
where the product over $N$ retains the original index order.
By Lemma~\ref{lem:sm-atomic-dressed-string}, each dressed factor lies
in the Gaussian cone supported on $S_j$.  The unused-site density
$\mathcal D_{S_0}^{2}$ is a disjoint product of one-site cone elements.
Supported disjoint-product closure and nonnegative scaling by $c_E$
therefore prove~\eqref{eq:sm-terminal-sandwich}.
If $c_E=0$, the entire operator is zero; if $N=\varnothing$, the
right-hand side is simply $c_E\mathcal D_\Lambda^2$.
An empty terminal family is included by the empty-product convention
$c_E=1$.  Thus no nonempty-string or nonempty-family assumption is needed.
\end{proof}

\subsection{Recordwise site decoupling and exact terminal expansion}
The following theorem supplies the exact Gibbs-operator identity, together
with the support and convergence properties needed to apply the atomic cone
lemmas. Its proof spells out the mapped support-decoupling construction
for this interaction picture.

\begin{thm}[Exact terminal site-pinning representation]
\label{thm:sm-terminal-site-expansion}
Let $\beta\ge0$, $\tau_\beta=\beta/2$, and
$A_\beta=e^{-\beta UV_{\rm int}/2}$. Let $\mathcal L_{\mathrm{str}}$ be the finite
active virtual-label set of the atomic interaction-picture expansion of
$H_{U,t}$, and let $\Delta_{\mathrm{site}}$ be its integrated site-conflict
strength, including self-conflicts. If
$\Delta_{\mathrm{site}}\le1/72$, there exist an integer
$m\le N:=|\mathcal L_{\mathrm{str}}|$, a countable outcome set $\Omega$, and numbers
$p_\omega\ge0$ with $\sum_\omega p_\omega=1$, together with real
coefficients $\alpha_{\omega j}$ and even Majorana supports
$J_{\omega j}$, such that
\[
 |\alpha_{\omega j}|\le1,\qquad
 \alpha_{\omega j}=0\ \Longrightarrow\ J_{\omega j}=\varnothing,
\]
the site supports $\pi(J_{\omega j})$ are pairwise disjoint, and
\begin{equation}\label{eq:sm-terminal-gibbs}
 e^{-\beta H_{U,t}}
 =\sum_{\omega\in\Omega}p_\omega A_\beta
 \left(\prod_{j=1}^{m}
 (I+\alpha_{\omega j}\gamma_{J_{\omega j}})\right)A_\beta.
\end{equation}
The operator-valued sum is absolutely convergent in norm. Its factors are
kept in the displayed order throughout the construction, and a
zero-coefficient factor is exactly the identity with empty support.
\end{thm}

\begin{proof}
We first treat $\beta>0$ and $\Delta_{\mathrm{site}}>0$.
For $S\subseteq\mathcal L_{\mathrm{str}}$, let
$V_S(s)=\sum_{\lambda\in S}z_\lambda(s)P_\lambda$ and denote its
time-ordered propagator by $\widehat{U}_S(s,u)$, with the convention of
the interaction-picture construction above. Define
\[
 \Phi^{\mathrm{site}}_S(K)=\frac1{\Delta_{\mathrm{site}}}
 \sum_{\substack{\lambda\in S:\
       \pi(\operatorname{supp}P_\lambda)\cap\pi(K)\ne\varnothing}}
 a_\lambda.
\]

\paragraph{1. Initial exact identity and induction invariant.}
A record at stage $k$ consists of a remaining set $S_\omega$ and an
ordered list of exactly $k$ affine factors. Put
\[
 B_\omega=\prod_{j=1}^{k}
 (I+\alpha_{\omega j}\gamma_{J_{\omega j}}),\qquad
 F(S,B)=A_\beta \widehat{U}_S(\tau_\beta,0)
 B \widehat{U}_S(0,-\tau_\beta)A_\beta.
\]
The exact virtual expansion and symmetric interaction-picture identity
give
\[
 F(\mathcal L_{\mathrm{str}},I)
 =A_\beta \widehat{U}_{\mathcal L_{\mathrm{str}}}(\tau_\beta,0)
   \widehat{U}_{\mathcal L_{\mathrm{str}}}(0,-\tau_\beta)A_\beta
 =e^{-\beta H_{U,t}}.
\]
Thus the singleton record with probability one, remaining set
$\mathcal L_{\mathrm{str}}$, and an empty factor list has the required initial
expectation. A factor is called active when its coefficient is nonzero.
We maintain the following assertions at every stage:
\begin{enumerate}
\item The outcomes form a countable probability distribution, and the
family $p_\omega F(S_\omega,B_\omega)$ is absolutely norm summable
with sum $e^{-\beta H_{U,t}}$.
\item Each factor has a real coefficient and a normalized Hermitian
even Majorana monomial. A zero coefficient is stored with empty support
and identity monomial.
\item Distinct active factors have disjoint site supports.
\item Every active factor strictly before the last factor is finished:
$\Phi^{\mathrm{site}}_{S_\omega}(J_{\omega j})=0$ for $j<k$.
\item Every coefficient obeys
$|\alpha_{\omega j}|\le2^{-\Phi^{\mathrm{site}}_{S_\omega}(J_{\omega j})}$.
\item Every remaining set has cardinality $N-k$.
\end{enumerate}
The factor assertions are vacuous initially. Both the remaining set and
the future deletion root are allowed to depend on the current record.

\paragraph{2. Finished-prefix preparation and root selection.}
Fix a record with $S\ne\varnothing$. If its last factor is active and
has nonzero remaining potential, denote it by $Y=I+\alpha X$ and select
$r\in S$ conflicting with its site support. Retain the preceding
factors and replace the old last slot by the inactive identity. If the
last factor is inactive or already finished, retain the whole old list,
put $\alpha=0$, $X=I$, and $Y=I$, and choose any $r\in S$. The same prescription applies to an
empty list. In both cases the old product is exactly $B=CY$, where the
retained list $C$ has the old length and all its active factors are
finished. Its active supports are disjoint from the support of $Y$.
Write $\widehat S=S\setminus\{r\}$ and $\theta=a_r/\Delta_{\mathrm{site}}$.

Every virtual label has positive activity at positive time radius.
Indeed, its continuous coefficient is nonzero somewhere on
$[-\tau_\beta,\tau_\beta]$; time reflection moves such a value to
$[0,\tau_\beta]$, where continuity gives a strictly positive integral
of its absolute value. Self-conflict gives $a_r\le\Delta_{\mathrm{site}}$, so
$0<\theta\le1$. Whenever $\alpha\ne0$, the selected root has the exact
deletion drop
\begin{equation}\label{eq:sm-terminal-potential-drop}
 \Phi^{\mathrm{site}}_S(\operatorname{supp}X)
 =\Phi^{\mathrm{site}}_{\widehat S}(\operatorname{supp}X)+\theta.
\end{equation}
No positive-drop assertion is needed when $\alpha=0$.

A finished active factor is site-disjoint from every label in $S$:
its potential is a sum of nonnegative terms and all these label
activities are positive. It is therefore also disjoint in intrinsic
Majorana support. Even CAR monomials on disjoint supports commute.
Consequently $C$ commutes with the remaining propagators and with the
exact deletion bridge
\[
 R_{\mathrm{str}}=\widehat{U}_{\widehat S}(0,\tau_\beta)
       \widehat{U}_S(\tau_\beta,0).
\]

\paragraph{3. Connected sampling and the seven affine updates.}
Theorem~\ref{thm:sm-site-connected-propagator} supplies a countable
distribution $q_a$ and operators $Z_a=I+b_aE_a$ such that
\[
 q_a\ge0,\qquad \sum_aq_a=1,\qquad \sum_aq_aZ_a=R_{\mathrm{str}},
\]
with a norm-summable expectation. The residual outcome has amplitude
zero, empty support, and length zero. A nonresidual outcome is an
integrated history of integer length $\ell_a\ge1$, with a unit-modulus
phased even Majorana monomial $E_a$ and
\begin{equation}\label{eq:sm-terminal-history-budget}
 |b_a|\le\theta(3\Delta_{\mathrm{site}})^{\ell_a},\qquad
 \Phi^{\mathrm{site}}_{\widehat S}(\operatorname{supp}E_a)\le\ell_a.
\end{equation}
For a nonzero amplitude, the intrinsic support is the symmetric difference
of its virtual-label supports. Zero amplitudes may use the identity monomial.
In either case the site image is contained in the union of the history supports. All history
labels lie in $S$. Thus retained active supports are disjoint from all
history supports, and the potential of the union of two history
supports is at most $\ell_a+\ell_b$.

Draw two such outcomes $a,b$ independently and select one of seven
choices independently, with probabilities
\[
 w_0=2^{-\theta},\qquad w_1=\cdots=w_6=w=
 \frac{1-2^{-\theta}}6.
\]
They are strictly positive and sum to one. With
$\mathcal H(E)=(E+E^*)/2$, define the seven Hermitian updates
\begin{align*}
 D_0&=\alpha X,&
 D_1&=b_a\mathcal H(E_a),&
 D_2&=b_b\mathcal H(E_b),\\
 D_3&=\alpha b_a\mathcal H(E_aX),&
 D_4&=\alpha b_b\mathcal H(E_bX),\\
 D_5&=b_ab_b\mathcal H(E_aE_b^*),&
 D_6&=\alpha b_ab_b\mathcal H(E_aXE_b^*).
\end{align*}
For choice $c$, the new affine operator is $Y_{ab,c}=I+D_c/w_c$.
Products of phased even Majorana monomials remain phased even Majorana
monomials; their Hermitian parts multiply normalized Hermitian
monomials by real numbers of absolute value at most one. Hence
$Y_{ab,c}=I+\alpha'\gamma_{J'}$ exactly, with real $\alpha'$ and even
$J'$. If $\alpha'=0$, use $J'=\varnothing$. Its support lies in the
appropriate union of the old support and the one or two sampled
history supports. Append this factor to $C$ and use remaining set
$\widehat S$. The successor has exactly $k+1$ stored factors: in the
unfinished-last case the old last slot remains as an inactive identity.

\paragraph{4. Support and coefficient bounds for every choice.}
The retained factors remain finished because
$0\le\Phi^{\mathrm{site}}_{\widehat S}\le\Phi^{\mathrm{site}}_S$. They are disjoint from both the old
factor and all history supports, and hence from the new support. Thus
the support, normalization, and finished-prefix assertions persist.
For the retain choice, the coefficient bound is immediate if $\alpha=0$.
Otherwise, Eq.~\eqref{eq:sm-terminal-potential-drop} pays for the inverse
probability exactly:
\[
 \frac{|\alpha|}{w_0}
 \le2^{-\Phi^{\mathrm{site}}_{\widehat S}(\operatorname{supp}X)}.
\]
For a nontrivial choice, concavity on $[0,1]$ gives
$1-2^{-\theta}\ge\theta/2$, so $\theta/w\le12$. Since
$\Delta_{\mathrm{site}}\le1/72$, every integer $\ell\ge1$ satisfies
\[
 \frac{\theta(3\Delta_{\mathrm{site}})^\ell}{w}
 \le12(3\Delta_{\mathrm{site}})^\ell\le2^{-\ell}.
\]
For two nonresidual histories, $\theta^2\le\theta$ similarly gives
\[
 \frac{\theta^2(3\Delta_{\mathrm{site}})^{\ell_a+\ell_b}}w
 \le2^{-(\ell_a+\ell_b)}.
\]
An update containing a residual amplitude is zero. Equation
\eqref{eq:sm-terminal-history-budget} converts these length bounds to
the corresponding potential bounds. A mixed update includes the
additional factor
$|\alpha|\le2^{-\Phi^{\mathrm{site}}_{\widehat S}(\operatorname{supp}X)}$.
Subadditivity of potential on unions, followed by monotonicity on
subsets, therefore gives
$|\alpha'|\le2^{-\Phi^{\mathrm{site}}_{\widehat S}(J')}$ in all seven cases. This
preserves the full coefficient invariant needed at the next step, not
only its weaker consequence $|\alpha'|\le1$.

\paragraph{5. Exact conditional operator expectation.}
Expanding the seven updates, using $X=X^*$ but making no commutation
assumption about the sampled operators, gives
\[
 \sum_cw_cY_{ab,c}
 =\frac12\bigl(Z_aYZ_b^*+Z_bYZ_a^*\bigr).
\]
The two samples are independent, so their subsequent average is
\begin{equation}\label{eq:sm-terminal-fiber-expectation}
 \sum_{a,b,c}q_aq_bw_cY_{ab,c}=R_{\mathrm{str}}Y R_{\mathrm{str}}^*.
\end{equation}
All the sums here are justified in norm. The sampler expectations are
absolutely summable in the finite-dimensional operator space, so the
independent product family is summable. Moreover the packaged new
factors satisfy $\|Y_{ab,c}\|\le2$, and their weighted norms are bounded
by the summable family $2q_aq_bw_c$. Thus the finite choice sum and the
countable sample sums may be interchanged.

The exact consecutive-deletion identities and time reflection yield
\[
 \widehat{U}_S(\tau_\beta,0)
 =\widehat{U}_{\widehat S}(\tau_\beta,0)R_{\mathrm{str}},\qquad
 \widehat{U}_S(0,-\tau_\beta)
 =R_{\mathrm{str}}^*\widehat{U}_{\widehat S}(0,-\tau_\beta).
\]
Multiply Eq.~\eqref{eq:sm-terminal-fiber-expectation} on the left by
$A_\beta \widehat{U}_{\widehat S}(\tau_\beta,0)C$ and on the right by
$\widehat{U}_{\widehat S}(0,-\tau_\beta)A_\beta$. The established
commutation $C R_{\mathrm{str}}=R_{\mathrm{str}}C$ gives
\begin{align}
 \sum_{a,b,c}q_aq_bw_cF(\widehat S,CY_{ab,c})
 =A_\beta \widehat{U}_{\widehat S}(\tau_\beta,0)
 C R_{\mathrm{str}}Y R_{\mathrm{str}}^*\widehat{U}_{\widehat S}(0,-\tau_\beta)A_\beta
 =F(S,CY).
 \label{eq:sm-terminal-conditional-expectation}
\end{align}
In particular, the unfinished factor has not been commuted through the
deletion bridge.

\paragraph{6. Probability normalization and the countable tower identity.}
For an old outcome $\omega$, the successor probability is
\[
 p'_{\omega,a,b,c}=p_\omega
 q_{\omega,a}q_{\omega,b}w_{\omega,c}.
\]
The outcome space is the countable disjoint union of these fibers;
each root, sample distribution, and choice distribution may depend on
$\omega$. Nonnegativity and normalization follow from those of the
old distribution and each conditional fiber.

To justify summing Eq.~\eqref{eq:sm-terminal-conditional-expectation}
over old outcomes, define the finite uniform budget at factor count $k$,
\[
 C_k=\|A_\beta\|^2\,2^k
 \left(\sum_{S\subseteq\mathcal L_{\mathrm{str}}}
 \|\widehat{U}_S(\tau_\beta,0)\|\right)
 \left(\sum_{S\subseteq\mathcal L_{\mathrm{str}}}
 \|\widehat{U}_S(0,-\tau_\beta)\|\right).
\]
The two sums are finite because $\mathcal L_{\mathrm{str}}$ is finite. Nonnegative
potential and the coefficient invariant imply $|\alpha_j|\le1$.
Each affine factor therefore has norm at most two and its ordered
product has norm at most $2^k$. Submultiplicativity proves
$\|F(S,B)\|\le C_k$ uniformly in the record. In particular,
\[
 \sum_{\omega'}\|p'_{\omega'}
 F(S_{\omega'},B_{\omega'})\|
 \le C_{k+1}\sum_{\omega'}p'_{\omega'}=C_{k+1}<\infty.
\]
Multiplying the conditional identity by $p_\omega$ and summing over
old outcomes is consequently legitimate, and preserves the exact
Gibbs expectation. This establishes absolute summability at every
stage, rather than relying on a formal rearrangement of countable
sums.

\paragraph{7. Finite termination and collapse to the terminal sum.}
Every successor deletes a root belonging to its own remaining set.
Consequently its remaining cardinality decreases by exactly one. At
stage $k$, all records have $N-k$ remaining labels, independently of
their possibly different deletion orders. After $N$ steps every
remaining set is empty and every stored list has length $N$.
The remaining generator is then zero, both remaining propagators are
the identity, and every remaining potential is zero. The coefficient
invariant gives $|\alpha_{\omega j}|\le1$, while the exact expectation
invariant becomes Eq.~\eqref{eq:sm-terminal-gibbs}. Its norm sum is
at most $\|A_\beta\|^2 2^N$. Inactive supports are empty, so
disjointness of active site supports gives pairwise disjointness of
the entire stored support list.

\paragraph{8. Zero-time and zero-conflict cases.}
If $\beta=0$, choose directly a singleton terminal distribution with
remaining set $\varnothing$ and an empty factor list. Equal-time
propagators and $A_\beta$ are the identity, and this representation
gives $e^0=I$ with $m=0$. It need not coincide with the canonical
all-label initial record: a coefficient can be nonzero at the single
point zero, even though its integrated activity is zero.

If $\beta>0$ and $\Delta_{\mathrm{site}}=0$, any active virtual label
would have strictly positive activity by the continuity and reflection
argument in step 2. Self-conflict would simultaneously imply
$a_\lambda\le\Delta_{\mathrm{site}}=0$, a contradiction. Hence
$\mathcal L_{\mathrm{str}}=\varnothing$, and the canonical initial singleton itself
is terminal, again with $m=0$. Neither boundary case divides by
$\Delta_{\mathrm{site}}$ or by a vanishing branch probability.
\end{proof}

\subsection{Completion of the strong-coupling proof}
\begin{proof}[Proof of Theorem~\ref{thm:hubbard-strong-sm}]
By Proposition~\ref{prop:sm-site-overlap-bound} and the assumed condition,
\[
 \Delta_{\mathrm{site}}
 \le\frac{8D|t|}{|U|}\left(e^{\beta|U|/2}-1\right)
 \le\frac1{72}.
\]
Theorem~\ref{thm:sm-terminal-site-expansion} therefore gives the exact,
absolutely norm-convergent expansion~\eqref{eq:sm-terminal-gibbs}.
For each outcome, Corollary~\ref{cor:sm-terminal-atomic-cone} places its
site-disjoint terminal sandwich in $\FGC$. Since $p_\omega\ge0$,
closure of $\FGC$ under nonnegative scaling and norm-convergent countable
sums gives $e^{-\beta H_{U,t}}\in\FGC$. This invokes the atomic dressed-factor
lemma, not Gaussian congruence by the quartic atomic half-Gibbs factor.
The Hermitian Hamiltonian has a positive-definite Gibbs operator and hence
a strictly positive trace. Dividing by this trace gives the convex-Gaussian
state $\rho_\beta(H_{U,t})$.

Finally, when $t\ne0$ and $D>0$, multiplication by the positive quantity
$|U|/(8D|t|)$ and monotonicity of the exponential give
\[
 \frac{8D|t|}{|U|}\left(e^{\beta|U|/2}-1\right)\le\frac1{72}
 \quad\Longleftrightarrow\quad
 \beta\le\frac2{|U|}\log\!\left(1+\frac{|U|}{576D|t|}\right).
\]
This proves the stated inverse-temperature window as well.
\end{proof}

\end{document}